\documentclass[sigconf,nonacm,usenames,dvipsnames,table]{acmart} 
\pdfoutput=1
\makeatletter        
\def\mdseries@tt{m} 
\makeatother
\usepackage[draft=true]{minted}

\usepackage[utf8]{inputenc}
\usepackage{endnotes}
\usepackage{comment}
\usepackage{graphicx}
\usepackage{url}
\usepackage{color}
\usepackage{minted}
\usepackage{listings}
\usepackage{float}
\usepackage{lipsum}
\usepackage{enumitem}
\usepackage{tabularx}
\usepackage{makecell}
\usepackage{tikz}
\usetikzlibrary{backgrounds,fit, decorations.markings, patterns}
\usepackage{pgfplots}
\pgfplotsset{compat=1.5}
\usepackage{amssymb,amsmath}
\usepackage{ifthen}
\usepackage{soul}
\usepackage{seqsplit}
\usepackage{xspace}
\usepackage{framed}
\usepackage{amsthm}
\usepackage{glossaries-extra}
\usepackage{subcaption}

\includecomment{extended}
\excludecomment{conference}

\makeatletter
\@ifundefined{c@rownum}{%
  \let\c@rownum\rownum
}{}
\@ifundefined{therownum}{%
  \def\therownum{\@arabic\rownum}%
}{}
\makeatother

\expandafter\def\expandafter\UrlBreaks\expandafter{\UrlBreaks
  \do\a\do\b\do\c\do\d\do\e\do\f\do\g\do\h\do\i\do\j%
  \do\k\do\l\do\m\do\n\do\o\do\p\do\q\do\r\do\s\do\t%
  \do\u\do\v\do\w\do\x\do\y\do\z\do\A\do\B\do\C\do\D%
  \do\E\do\F\do\G\do\H\do\I\do\J\do\K\do\L\do\M\do\N%
  \do\O\do\P\do\Q\do\R\do\S\do\T\do\U\do\V\do\W\do\X%
  \do\Y\do\Z%
}

\theoremstyle{definition}
\newtheorem{defn}{Definition}

\newcommand{\tool}{AC\-Miner\xspace}
\newcommand{\cp}{control predicate\xspace}
\newcommand{\cps}{control predicates\xspace}
\newcommand{\CP}{Control Predicate\xspace}
\newcommand{\CPS}{Control Predicates\xspace}
\newcommand{\cq}{context query\xspace}
\newcommand{\cqs}{context queries\xspace}
\newcommand{\CQ}{Context Query\xspace}
\newcommand{\CQS}{Context Queries\xspace}

\makeglossaries
\glssetcategoryattribute{common}{nohyper}{true}
\newglossaryentry{target_p1} {
  name={target entry point},
  plural={target entry points},
  description={},
  category={common}
}
\newglossaryentry{support_p1} {
  name={supporting entry point},
  plural={supporting entry points},
  description={},
  category={common}
}
\newglossaryentry{supportac_p1} {
  name={supporting authorization check},
  plural={supporting authorization checks},
  description={},
  category={common}
}
\newglossaryentry{recommend_p1} {
  name={recommended authorization check},
  plural={recommended authorization checks},
  description={},
  category={common}
}
\newglossaryentry{arule_p1} {
  name={association rule},
  plural={association rules},
  description={},
  category={common}
}
\newglossaryentry{fp_p1} {
  name={false positive},
  plural={false positives},
  description={},
  category={common}
}

\newcommand{\pInitialContextQueries}{33}
\newcommand{\pManualContextQueries}{620}
\newcommand{\pTotalContextQueries}{875}
\FPeval{\pPrecentIdContextQueries}{round((\pManualContextQueries/\pTotalContextQueries)*100,0)}
\FPeval{\pTotalDiffManualContextQueries}{clip(\pTotalContextQueries{}-\pManualContextQueries{})}
\FPeval{\pPrecentIncreaseContextQueries}{round(((\pTotalContextQueries{}-\pInitialContextQueries{})/\pInitialContextQueries{})*100,0)}
\newcommand{\pInitialControlPredicates}{25808}
\newcommand{\pTotalControlPredicates}{3308}

\newcommand{\pHighRiskHighImpactEp}{7}
\newcommand{\pHighRiskHighImpactRule}{7}
\newcommand{\pHighRiskEp}{1}
\newcommand{\pHighRiskRule}{1}
\newcommand{\pHighImpactEp}{12}
\newcommand{\pHighImpactRule}{14}
\newcommand{\pLowRiskLowImpactEp}{8}
\newcommand{\pLowRiskLowImpactRule}{8}
\FPeval{\pTotalRiskImpactEp}{clip(\pHighRiskHighImpactEp{}+\pHighRiskEp{}+\pHighImpactEp{}+\pLowRiskLowImpactEp{})}
\FPeval{\pTotalRiskImpactRule}{clip(\pHighRiskHighImpactRule{}+\pHighRiskRule{}+\pHighImpactRule{}+\pLowRiskLowImpactRule{})}

\newcommand{\pUserEps}{14}

\newcommand{\pOtherEps}{11}

\newcommand{\pEpTotal}{4004}
\newcommand{\pEpWithLogic}{1995}
\newcommand{\pEpWithAtLeastOneRule}{246}
\FPeval{\pEpReduction}{round(((\pEpTotal - \pEpWithAtLeastOneRule) / \pEpTotal)*100,0)}

\newcommand{\pRulesDiffInFunctionality}{189}
\newcommand{\pRulesDiffArgs}{66}
\newcommand{\pRulesNoise}{53}
\newcommand{\pRulesSpecialCaller}{37}
\newcommand{\pRulesPermissionCheck}{23}
\newcommand{\pRulesEquivalentChecks}{10}
\newcommand{\pRulesOther}{20}
\FPeval{\pTotalFalsePostives}{clip(\pRulesDiffInFunctionality{}+\pRulesDiffArgs{}+\pRulesNoise{}+\pRulesSpecialCaller{}+\pRulesPermissionCheck{}+\pRulesEquivalentChecks{}+\pRulesOther{})}
\newcommand{\pRulesSpeedUp}{7}
\newcommand{\pRulesIncreasedAccess}{2}
\newcommand{\pRulesUnusedContextQueries}{16}
\FPeval{\pTotalNonVulnerability}{clip(\pRulesSpeedUp{}+\pRulesIncreasedAccess{}+\pRulesUnusedContextQueries{})}
\FPeval{\pTotalOtherInconsistencies}{clip(\pTotalFalsePostives{}+\pTotalNonVulnerability{})}
\newcommand{\pTotalAssociationRules}{453}

\newcommand{\pSubmittedBugReportsMedium}{2}

\newboolean{showcomments}
\setboolean{showcomments}{true}
\ifthenelse{\boolean{showcomments}}
{ \newcommand{\mynote}[2]{
    \fbox{\bfseries\sffamily\scriptsize#1}
    {\small$\blacktriangleright$\textsf{\emph{#2}}$\blacktriangleleft$}}}
{ \newcommand{\mynote}[2]{}}

\newcommand{\myp}[1]{\vspace{0.1em}\noindent\textbf{#1:}}
\newcommand{\emparagraph}[1]{\vspace{0.1em}\textit{#1:}}

\newcommand\lword[1]{\leavevmode\nobreak\hskip0pt plus\linewidth\penalty50\hskip0pt plus-\linewidth\nobreak#1}

\newcommand{\ttt}[1]{%
  \begingroup
    \protect\renewcommand{\seqinsert}{\ifmmode\allowbreak\else\-\fi}%
    \protect\texttt{\protect\seqinsert{\protect\seqsplit{\small#1}}}%
  \endgroup
}
\newcommand{\tttscript}[1]{%
  \begingroup
    \protect\renewcommand{\seqinsert}{\ifmmode\allowbreak\else\-\fi}%
    \protect\texttt{\protect\seqinsert{\protect\seqsplit{\scriptsize#1}}}%
  \endgroup
}
\newcommand{\tttfoot}[1]{%
  \begingroup
    \protect\renewcommand{\seqinsert}{\ifmmode\allowbreak\else\-\fi}%
    \protect\texttt{\protect\seqinsert{\protect\seqsplit{\footnotesize#1}}}%
  \endgroup
}

\definecolor{darkgreen}{RGB}{0,102,0}
\definecolor{darkorange}{RGB}{255, 102, 0}

\lstdefinestyle{javaStyle} {
  language=Java,
  showspaces=false,
  showtabs=false,
  breaklines=false,
  showstringspaces=false,
  breakatwhitespace=true,
  numbers=left,
  numberstyle=\scriptsize,
  tabsize=2,
  captionpos=b,
  commentstyle=\bfseries\color{gray},
  keywordstyle=\bfseries\color{Plum},
  stringstyle=\color{red}\bfseries,
  basicstyle=\ttfamily\footnotesize,
  moredelim=[il][\textcolor{lightgray}]{\$\$},
  moredelim=[is][\textcolor{lightgray}]{\%\%}{\%\%}
}

\begin{conference}
\copyrightyear{2019} 
\acmYear{2019} 
\setcopyright{acmlicensed}
\acmConference[CODASPY '19] {Ninth ACM Conference on Data and Application Security and Privacy}{March 25--27, 2019}{Richardson, TX, USA}
\acmBooktitle{Ninth ACM Conference on Data and Application Security and Privacy (CODASPY '19), March 25--27, 2019, Richardson, TX, USA}
\acmPrice{15.00}
\acmDOI{10.1145/3292006.3300023}
\acmISBN{978-1-4503-6099-9/19/3}
\end{conference}

\begin{extended}
\end{extended}

\newcommand{\class}[1]{\ttt{#1}}
\newcommand{\method}[1]{\ttt{#1}}

\begin{document}
\title{\tool: Extraction and Analysis of Authorization Checks in Android's Middleware}

\author{Sigmund Albert Gorski III}
\affiliation{%
  \institution{North Carolina State University}
}
\email{sagorski@ncsu.edu}

\author{Benjamin Andow}
\affiliation{%
  \institution{North Carolina State University}
}
\email{beandow@ncsu.edu}

\author{Adwait Nadkarni}
\affiliation{%
  \institution{William \& Mary}
}
\email{nadkarni@cs.wm.edu}

\author{Sunil Manandhar}
\affiliation{%
  \institution{William \& Mary}
}
\email{sunil@cs.wm.edu}

\author{William Enck}
\affiliation{%
  \institution{North Carolina State University}
}
\email{whenck@ncsu.edu}

\author{Eric Bodden}
\affiliation{%
  \institution{Paderborn University}
}
\email{eric.bodden@uni-paderborn.de}

\author{Alexandre Bartel}
\affiliation{%
  \institution{University of Luxembourg}
}
\email{alexandre.bartel@uni.lu}

\renewcommand{\shortauthors}{Gorski et al.}

\begin{abstract}
Billions of users rely on the security of the Android platform to protect phones, tablets, and many different types of consumer electronics.
While Android's permission model is well studied, the \emph{enforcement} of the protection policy has received relatively little attention.
Much of this enforcement is spread across system services, taking the form of hard-coded checks within their implementations.
In this paper, we propose Authorization Check Miner (\tool), a framework for evaluating the correctness of Android's access control enforcement through consistency analysis of authorization checks.
\tool combines program and text analysis techniques to generate a rich set of authorization checks, mines the corresponding protection policy for each service entry point, and uses association rule mining at a service granularity to identify inconsistencies that may correspond to vulnerabilities.
We used \tool to study the AOSP version of Android 7.1.1 to identify \pTotalRiskImpactEp{} vulnerabilities relating to missing authorization checks.
In doing so, we demonstrate \tool's ability to help domain experts process thousands of authorization checks scattered across millions of lines of code.
\end{abstract}

\begin{extended}
\thanks{This is the extended version of the \tool paper published in the proceedings of the ninth ACM Conference on Data and Application Security and Privacy (CODASPY) 2019.}
\end{extended}

\maketitle

\section{Introduction}

Android has become the world's dominant computing platform, powering over 2 billion devices by mid-2017~\cite{android_stats}.
Not only is Android the primary computing platform for many end-users, it also has significant use by business enterprises~\cite{android_business} and government agencies~\cite{android_governmenta,android_governmentb}. 
As a result, any security flaw in the Android platform is likely to cause significant and widespread damage, lending immense importance to evaluating the platform's security.

While Android is built on Linux, it has many differences.
A key appeal of the platform is its semantically rich application programming interfaces (APIs) that provide application developers simple and convenient abstractions to access information and resources (e.g., retrieve the GPS location, record audio using the microphone, and take a picture with the camera).
This functionality, along with corresponding security checks, is implemented within a collection of privileged userspace services.
%
While most Android security research has focused on third party applications~\cite{eom09b,fch+11,fwm+11,egc+10,pfnw12,ne13,sbl+14,bbh+15}, the several efforts that consider platform security highlight the need for more systematic evaluation of security and access control checks within privileged userspace services (e.g., evidence of system apps re-exposing information without security checks~\cite{gzwj12, wgz+13}, or missing checks in the Package Manager service leading to Pile-Up vulnerabilities~\cite{xpw+14}). 

To date, only two prior works have attempted to evaluate the correctness of access control logic within Android's system services. Both Kratos~\cite{socq+16} and AceDroid~\cite{ahs+18} approximate correctness through consistency measures, as previously done for evaluating correctness of security hooks in the Linux kernel~\cite{ejz02,jez04,tzm+08}. However, these prior works have limitations. Kratos only considers a small number of manually-defined authorization checks (e.g., it excludes App Ops checks). AceDroid considers a larger set of authorization checks, but these are still largely manually defined, primarily through observation. Kratos performs coarse-grained analysis using call-graphs, leading to imprecision. AceDroid's program analysis provides better precision, but oversimplifies its access control representation, making it difficult to identify vulnerabilities within single system images.

In this paper, we propose Authorization Check Miner (\tool), a framework for evaluating the correctness of Android's access control enforcement through consistency analysis of authorization checks.
\tool is based on several novel insights.
First, we avoid identification of protected operations (a key challenge in the space) by considering program logic between service entry points and code that throws a \ttt{SecurityException}.
Second, we propose a semi-automated method of discovering authorization checks.
More specifically, we mine all constants and names of methods and variables that influence conditional logic leading to throwing a \ttt{SecurityException}.
From this dataset, we identify security-relevant values (e.g., ``restricted'') and develop regular expressions to automatically identify those conditions during program analysis that mines policy rules from the code.
Third, we use association rule mining to identify inconsistent authorization checks for entry points in the same service.
Association rule mining has the added benefit of suggesting changes to make authorization checks more consistent, which is valuable when triaging results.
By applying this methodology, \tool allows a domain expert (i.e., a developer familiar with the AOSP source code) to quickly identify missing authorization checks that allow abuse by third-party applications.

We evaluated the utility of \tool by applying it to the AOSP code for Android 7.1.1.
Of the 4,004 total entry points to system services, \tool identified 1,995 with authorization checks. 
Of these entry points, the association rule mining identified inconsistencies in 246.
We manually investigated these 246 entry points with the aid of the rules suggested by the association rule mining, which allowed us to identify \pTotalRiskImpactEp{} security vulnerabilities. \tool not only reduced the effort required to analyze system services (i.e., by narrowing down to only 246 entry points out of 4004), but also allowed us to rapidly triage results by suggesting solutions. Out of the \pTotalRiskImpactEp{} security vulnerabilities,  \pHighRiskHighImpactEp{} were in security-sensitive entry points that may be exploited from third-party applications, and an additional \pHighImpactEp{} were in security-sensitive entry points that may be exploited from system applications. The rest were in entry points with relatively low security value.
All \pTotalRiskImpactEp{} vulnerabilities have been reported to Google.
At the time of writing, Google has confirmed \pSubmittedBugReportsMedium{} of these vulnerabilities as ``moderate severity.''

This paper makes the following contributions:
\begin{itemize}

    \item \textit{We design and implement \tool, a framework that enables a domain expert to identify and systematically evaluate inconsistent access control enforcement in Android's system services.}
    Our results show that this analysis is not only useful for identifying existing vulnerabilities, but also inconsistencies that may lead to vulnerabilities in the future.
    
    \item \textit{We combine program and text analysis techniques to generate a rich set of authorization checks used in system services.} This technique is a significant improvement over prior approaches that use manually-defined authorization checks.
    
    \item \textit{We use \tool to evaluate the AOSP version of Android 7.1.1 and identify \pTotalRiskImpactEp{} vulnerabilities.} All vulnerabilities have been reported to Google, which at the time of writing has classified \pSubmittedBugReportsMedium{} as ``moderate severity.''
    
\end{itemize}

\begin{extended}
We designed \tool to give security researchers and system developers deep insights into the security of Android’s system services. 
\end{extended}
This paper describes how \tool can systematically analyze the consistency of the authorization checks in the system services. 
However, \tool may also be useful for other forms of analysis.
For instance, \tool can aid regression testing, as it can be extended to highlight changes to the policy implementation on a semantic level. 
The information extracted by \tool can also be used to study the evolution of access control in Android, potentially discovering new vulnerabilities. 
Finally, since changes by OEMs have historically introduced vulnerabilities, OEMs can use \tool to validate their implemented checks against AOSP.

The remainder of this paper proceeds as follows. 
Section~\ref{sec:background} provides background. 
Section~\ref{sec:overview} describes the challenges and provides an overview of \tool. 
Section~\ref{sec:design} describes the design of \tool in detail. 
Sections~\ref{sec:eval} and~\ref{p1:sec:eval_otherInconsistencies} describe our analysis of the system services of AOSP 7.1.1. 
Section~\ref{p1:sec:limitations} describes the limitations of our approach. 
Section~\ref{sec:relwork} discusses related work. 
Section~\ref{sec:conc} concludes.

\section{Background and Motivation}
\label{sec:background}
\label{sec:motivation}

The Android middleware is implemented using the same component abstractions as third-party applications~\cite{eom09a}: activities, content providers, broadcast receivers, and services.
In this paper, we only consider service components, which provide daemon-like functionality.
Apps interface with service components via the \ttt{Binder} inter-process communication (IPC) mechanism, which consists of sending \textit{parcel} objects that indicate the target interface method being called via an integer. For the most part, Android's system services use the Android Interface Description Language (AIDL) to automatically generate the code that unmarshalles these parcels. Moreover, when interfacing with system services, third party apps rely on public APIs implementing a \textit{proxy} to construct the parcel.
When the parcel is unmarshalled by the service interface, the arguments are passed to a \textit{stub} that calls the corresponding \textit{entry point} method in the service component. 
\begin{extended}
While not all service entry points have corresponding public APIs, any third-party application can use reflection to invoke the entry points for ``hidden'' APIs.
\end{extended}

Android uses two broad techniques to enforce access control. 
For coarse-grained checks (i.e., at the component level), the Activity Manager Service (AMS) enforces policy specified in application manifest files.
This paper focuses on fine-grained checks (i.e., at the service entry point level), which are enforced using hard-coded logic within the service implementation.
This hard-coded logic includes variants of the \method{checkPermission} method, Unix Identifier (UID) checks, as well as many subtle checks based on service-specific state. 
Prior work~\cite{socq+16, ahs+18} has primarily relied on manual enumeration of these checks, which is error prone.
To simplify discussion in this paper, we refer to such methods that return or check Android system state as \textit{context queries.}

\begin{figure}
\centering
\resizebox{0.9\columnwidth}{!}{
\begin{tikzpicture}

\lstset{
    style=javaStyle,
    emph=[1]{hasBaseUserRestriction,isValidRestriction,checkManageUsersPermission,hasUserRestriction,isSameApp,hasManageUsersPermission,getCallingUid,checkComponentPermission},
    emphstyle=[1]{\color{blue}\bfseries},
    emph=[2]{String,SecurityException,UserRestrictionsUtils,UserHandle,Binder,ActivityManager,PackageManager,Process},
    emphstyle=[2]{\color{darkgreen}\bfseries},
    emph=[3]{key,userId,message,callingUid},
    emphstyle=[3]{\color{darkorange}\bfseries},
    emph=[4]{SYSTEM_UID,ROOT_UID,PERMISSION_GRANTED},
    emphstyle=[4]{\color{Brown}\bfseries},
}          

\node {
\begin{lstlisting}  % Start your code-block

// Entry point with correct authorization checks
boolean hasBaseUserRestriction(String key, int userId) {
  checkManageUsersPermission("hasBaseUserRestriction");
  // Unique check without a SecurityException
  if (!UserRestrictionsUtils.isValidRestriction(key))
    return false;
  ...}
  
// Entry point missing checkManageUsersPermission
boolean hasUserRestriction(String key, int userId) {
  if (!UserRestrictionsUtils.isValidRestriction(key))
    return false;
  ...}
  
void checkManageUsersPermission(String message) {
  if (!hasManageUsersPermission())
    throw new SecurityException();}

boolean hasManageUsersPermission() {
  int callingUid = Binder.getCallingUid();
  return UserHandle.isSameApp(callingUid, 
        Process.SYSTEM_UID)
    || callingUid == Process.ROOT_UID
    || ActivityManager.checkComponentPermission(
        "android.permission.MANAGE_USERS",
        callingUid, -1, true) == 
        PackageManager.PERMISSION_GRANTED;}

\end{lstlisting}
};
    
\end{tikzpicture}
}
\caption{Vulnerability found in \ttt{UserManagerService} by our tool}
\label{p1:fig:overview_code}
\end{figure}

Figure~\ref{p1:fig:overview_code} provides a motivating example for this paper, which contains a vulnerability discovered by \tool.
The simplified code snippet is from the User Manager Service, which provides core functionality similar to the Activity Manager and Package Manager Services.
In the figure, there are two entry points: \method{hasBaseUserRestriction} and \method{hasUserRestriction}.
The entry points perform very similar functionality, but have inconsistent authorization checks.
Specifically, \method{hasBaseUserRestriction} throws a \class{SecurityException} if the caller does not have the proper UID or permission.

This example is particularly apropos to \tool, because \method{hasUserRestriction} does not call any of the context queries considered by prior work~\cite{socq+16, ahs+18}. 
It also does not throw a \class{SecurityException}. 
Without knowledge that \method{isValidRestriction} is an authorization check, no form of consistency analysis could have identified that \method{hasUserRestriction} has a missing check.

\section{Overview}
\label{sec:overview}

The goal of this paper is to help a domain expert quickly identify and assess the impact of incorrect access control logic in implementations of system services in Android.
As with most nontrivial software systems, no ground truth specification of correctness exists.
Rather, the ``ground truth'' resides largely within the heads of the platform developers.
Prior literature has approached this type of problem by approximating correctness with consistency.
The intuition is that system developers are not malicious and that they are likely to get most of the checks correct.
The approach was first applied to security hook placement in the Linux kernel~\cite{ejz02,jez04,tzm+08} and more recently the Android platform~\cite{socq+16, ahs+18}.

Evaluating authorization check correctness via consistency analysis requires addressing the following challenges.
\begin{itemize}

\item \textit{Protected Operations:}
Nontrivial systems rarely have a clear specification of the functional operations that require protection by the access control system.
Protected operations range from accessing a device node to reading a value from a private member variable.
Axplorer~\cite{bbd+16} attempts to enumerate protected operations for Android; however, the specification is far from complete.

\item \textit{Authorization Checks:}
What constitutes an authorization check is vague and imprecise.
While some authorization checks are clear (e.g., those based on \class{checkPermission} and \class{getCallingUid}), many others are based on service-specific state and the corresponding authorization checks use a variety of method and variable names.

\item \textit{Consistency:}
The granularity and type of consistency impacts the precision and utility of the analysis.
While increasing the scope of relevant authorization checks increases the noise in the analysis,  not considering all authorization checks (as in Kratos~\cite{socq+16}) or using heuristics to determine relevancy (as in AceDroid~\cite{ahs+18}) raises the risk of not detecting vulnerabilities. 
\end{itemize}

\begin{figure}[t]
    \centering
    \scalebox{0.85}{
    \includegraphics[width=3.1in]{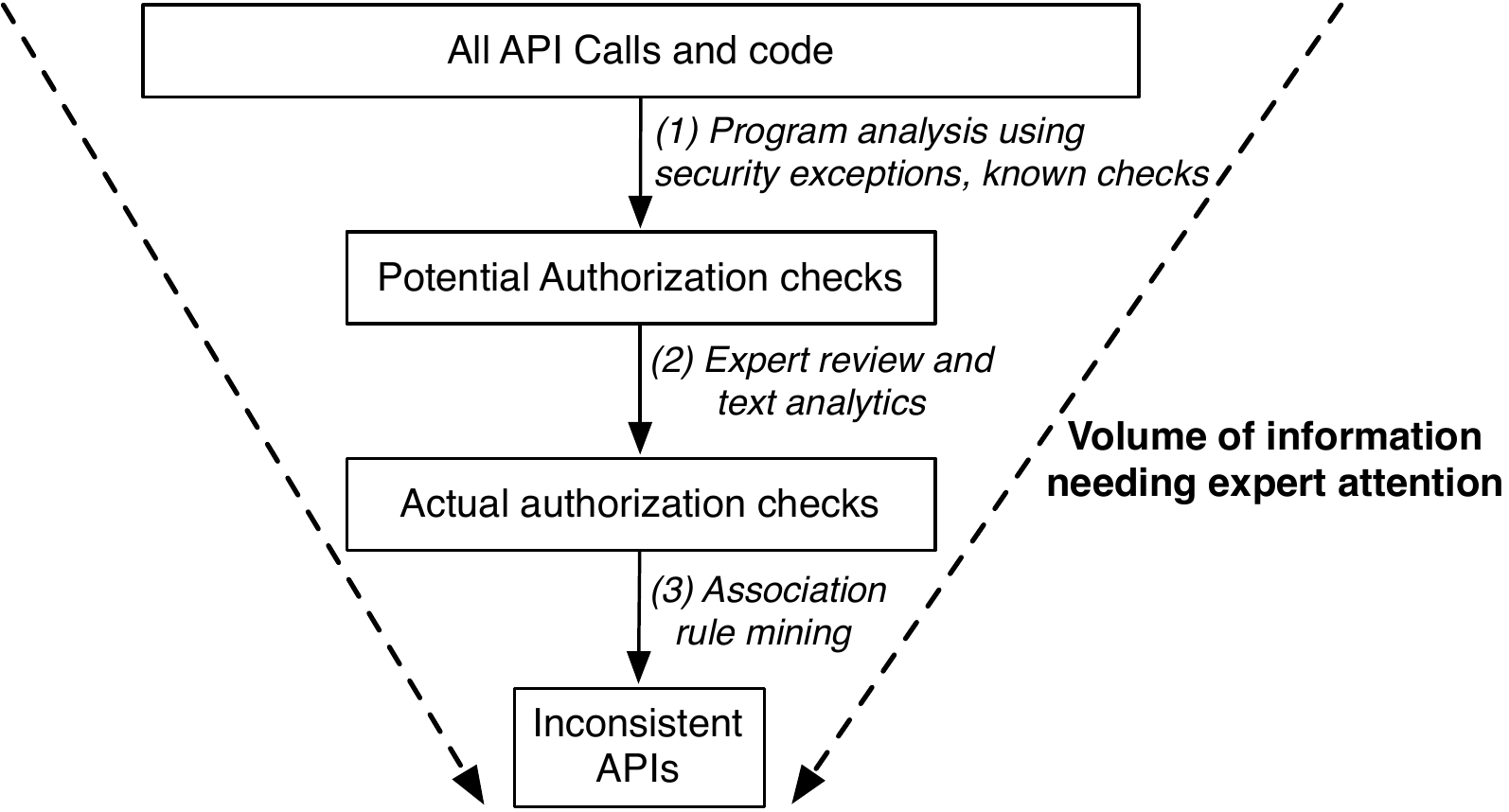}
    }
    \caption{Overview of \tool. At each stage, \tool significantly reduces the information an expert needs to analyze.}
    \label{fig:overview}
\end{figure}

\tool addresses these challenges through several novel insights.
First, \tool avoids the need to specify protected operations by considering program logic between service entry points and code that throws a \class{SecurityException}.
Our intuition is that if one control flow path leads to a \class{SecurityException}, an alternate control flow path leads to a protected operation.
Furthermore, the conditional logic leading to the \class{SecurityException} describes the authorization checks.
However, we found that not all authorization denials lead to a \class{SecurityException}, therefore, we also include entry points that contain known authorization checks.
Second, \tool semi-automatically discovers new authorization checks using a combination of static program analysis and textual processing.
More specifically, \tool identifies all of the method names, variable names, and strings that influence the conditional logic leading to a \class{SecurityException}.
The security-relevant values are manually refined and used to generate regular expressions that identify a broader set of authorization checks within service implementations.
Third and finally, \tool uses association rule mining for consistency analysis.
For each entry point, \tool uses static program analysis to extract a set of authorization checks.
Association rule mining compares the authorization check sets between entry points in the same service.
The analysis produces suggestions (called ``rules'') of how the sets should change to make them more consistent.
These rules include confidence scores that greatly aid domain experts when triaging the results.
This general approach is depicted in Figure~\ref{fig:overview}.

To more concretely understand how \tool operates, consider the discovery of the vulnerability shown in Figure~\ref{p1:fig:overview_code}. As a pre-processing step, \tool helps a domain expert semi-automatically identify authorization checks. First, \tool determines that the return value of \ttt{isValidRestriction} controls flow from the entry point \ttt{hasBaseUserRestriction} to a \ttt{SecurityException}. As such, this method name, along with many security irrelevant names are given to a domain expert. The domain expert then identifies security relevant terms (e.g. ``restriction``), which \tool consumes as part of a regular expression. Next, \tool mines the policy of the User Manager Service, extracting a policy for both \ttt{hasBaseUserRestriction} and \ttt{hasUserRestriction}, with the policy for \ttt{hasUserRestriction} only being extracted because \ttt{isValidRestriction} was identified as an authorization check. For each entry point, the policy is then encoded as a set of authorization checks (e.g. \ttt{isValidRestriction} and \ttt{ROOT\_UID} \ttt{==} \ttt{getCallingUid()}). Finally, \tool performs association rule mining to suggest set changes that make the policy more consistent. Such suggestions led us to discover the vulnerability in \method{hasUserRestriction}.

\section{Design}
\label{sec:design}

\tool is constructed on top of the Java static analysis framework Soot~\cite{vcg+99,lam2011soot} and has been largely parallelized so as to improve the run time of the complex analysis of Android's services.  
The design of \tool can be divided into three phases: (\S\ref{sec:mine-auth-check}) Mining Authorization Checks, (\S\ref{dsgn:refine}) Refining Authorization Checks, and (\S\ref{sec:const}) Inconsistency Analysis.

\subsection{Mining Authorization Checks}
\label{sec:mine-auth-check}

The first phase of \tool is mining authorization checks.
This phase consist of the following program analysis challenges: (\S\ref{sec:cg-const}) Call Graph Construction, (\S\ref{sec:id-auth-check}) Identifying Authorization-Check Statements, and (\S\ref{p1:sec:rep-auth-checks}) Representing Authorization Checks.

\begin{figure}[t]
    \centering
    \resizebox{0.9\columnwidth}{!}{
    \includegraphics{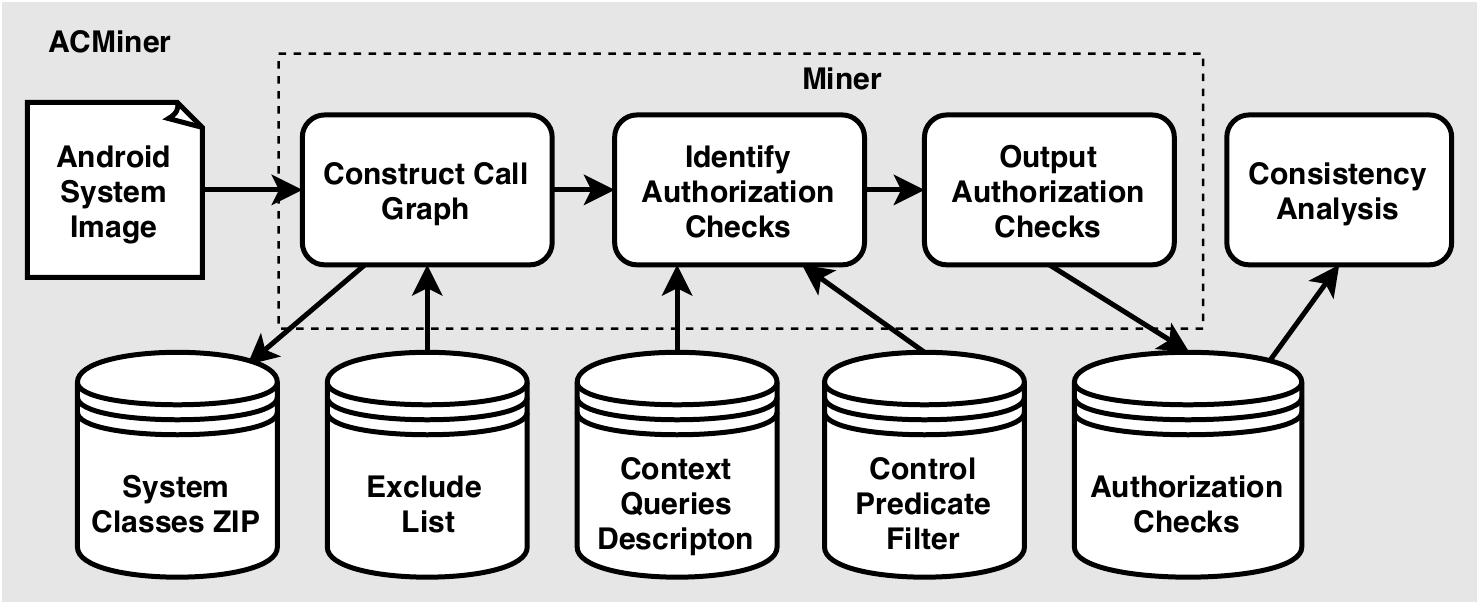}
    }
    \caption{\tool's processing stages and input files.}
    \label{fig:design_diagram}
\end{figure}

\subsubsection{Call Graph Construction}
\label{sec:cg-const}

Authorization checks are mined by traversing a call graph of the service implementation.
\tool constructs call graphs using the following three steps.

\myp{Extracting Java Class Files}
\tool extracts a \textit{.jar} containing all the class files of the Android middleware in Soot's Jimple format from Android's \textit{system.img}. This \textit{.jar} containing Jimple files is then used on all subsequent runs of \tool.
\begin{conference}
The implementation of this approach is detailed in the extended version~\cite{acminer-extended}.
\end{conference}
\begin{extended}
The implementation of this approach is detailed in Appendix~\ref{p1:app:extracting_java_class}.
\end{extended}

\myp{Extracting System Services and Entry Points}
Android's middleware is composed of isolated services that communicate through predefined \ttt{Binder} boundaries.
This division allows \tool to analyze each service separately, by defining each service by the code reachable through its \ttt{Binder} entry points. \tool extracts system services and their entry points similar to prior work~\cite{azhl12,socq+16,ahs+18}. 
\begin{conference}
For implementation details, please see the extended version~\cite{acminer-extended}.
\end{conference}
\begin{extended}
For implementation details, please see Appendix~\ref{p1:app:extracting_java_class_2}.
\end{extended}

\myp{Reducing the Call Graph}
\tool constructs a call graph representing all possible transitive calls from the entry points.
\tool uses the Class Hierarchy Analysis (CHA)~\cite{dgc95}, which is guaranteed to provide an over-approximation of the actual runtime call graph.
In contrast, Kratos and AceDroid use other potentially more accurate call graph builders (i.e., SPARK~\cite{lh03} and its WALA equivalent), which use points-to analysis to construct a less complete 
under-approximation of the runtime call graph.
The loss of completeness occurs when constructing call graphs for libraries and other Java code without main methods.
Therefore, it is important to note that unlike the prior work, \tool is more complete and guaranteed to include all paths containing authorization checks.

Since CHA call graphs are coarse over-approximations of the runtime call graph, \tool must apply heuristics to mitigate call graph bloat.
When resolving targets for method invocations, CHA considers every possible implementation of the target method whose declaring type is within the type hierarchy of the call's receiver type.
If the invoked method is defined in a widely used interface or superclass, the resolution may identify hundreds of targets for a single invocation.
Thus, the resulting CHA call graph for the Android middleware is too large to be analyzed in a reasonable amount of time and memory~\cite{lhotak2007comparing}.

To mitigate call graph bloat, we manually defined a list of classes and methods to exclude from the analysis, which become cutoff points in the call graph.
We ensured that the class or method subject to exclusion did not contain, lead to, or was used in an authorization check.
While the exclude list may require revision for newer versions of AOSP or modifications made by vendors, the creation of the exclude list is a largely one-time effort.
\begin{conference}
Please see the extended version~\cite{acminer-extended}
\end{conference}
\begin{extended}
Please see Appendix~\ref{p1:app:appendix_excludeList}
\end{extended}
for a detailed description of the exclusion procedure and our website~\cite{acminer-website} for a full list of excluded classes/methods.

Finally, when analyzing an entry point, \tool treats all other entry points as cutoff points in the call graph. This decision further reduces call graph bloat.
Unfortunately, it also introduces unsoundness into the call graph, which we discuss in Section~\ref{p1:sec:limitations}.

\subsubsection{Identifying Authorization Check Statements}
\label{sec:id-auth-check}

Once \tool has the call graph for all entry points, the next step is to identify authorization checks.
As described in Section~\ref{sec:motivation}, the complete set of authorization checks is unknown.
\tool takes a two pronged approach to identifying authorization checks.
First, it identifies all possible checks leading to a protected operation (this section).
Second, it refines the list of possible authorization checks based on code names and string values (Section~\ref{dsgn:refine}).
To describe this process, we must first define a \textit{\cp}.

\begin{defn}[\textit{\CP}]\label{defn:cps}
A conditional statement (i.e., an \ttt{if} or \ttt{switch} statement) whose logic authorizes access to some protected operation.
\end{defn}

Identifying protected operations is nontrivial, as they may range from accessing a device node to returning a private member variable.
However, even if we knew the protected operations, we would still need to determine which conditional statements are \cps.
\tool uses the key observation that Android frequently throws a \class{SecurityException} when access is denied.
Therefore, \tool marks all conditional statements on the control flow path between entry points and the statement throwing the \class{SecurityException} as \emph{potential} \cps.

While throwing a \class{SecurityException} is the most common way Android denies access, it is not the only way.
Some entry points deny access by returning false or even by returning empty data.
Such denials are not easily identifiable.
Fortunately, as shown in Figure~\ref{p1:fig:overview_code}, Android often groups authorization checks into methods to simplify placement. We refer to such groups of authorization checks as \textit{\cqs}.

\begin{defn}[\textit{\CQ}]\label{defn:cqs}
A method consisting entirely of a set of \textit{\cps}, calls to other \textit{\cqs}, and/or whose return value is frequently used as part of a \textit{\cp}.
\end{defn}

By identifying \textit{\cqs}, \tool can mark \cps not on the path between a entry point and a thrown \class{SecurityException}, thereby making the authorization check mining more complete.
As shown in Figure~\ref{fig:design_diagram}, \tool is configured with a input file that specifies \cqs.
Our method for defining this input is described in Section~\ref{dsgn:refine}.

Using these insights, \tool identifies authorization checks with fairly high accuracy.
The identification procedure is as follows.
First, \tool marks all conditional statements inside a \cq and the subsequent transitive callees as \cps for the entry point.
Next, \tool performs a backwards inter-procedural control flow analysis from each statement throwing a \class{SecurityException} and each \cq invocation to the entry point.
During this backwards analysis, \tool marks all conditional statements on the path as potential \cps.
Finally, to capture \cps that occur without a \class{SecurityException}, \tool performs a forward inter-procedural def-use analysis on the return value of a \cq.
During this analysis, \tool marks any conditional statement that uses the return value as a potential \cp.
Note, \tool does not currently track the return value through fields, as this was found to be too noisy.
However, \tool can track the return value through variable assignments, arithmetic operations, array assignments, and the passed parameters of a method invoke.

\subsubsection{Representing Authorization Checks}
\label{p1:sec:rep-auth-checks}

\tool represents the authorization checks for each entry point as a set of \cqs and \cps.
We initially represented authorization checks as boolean expressions representing the control flow decisions that lead to a thrown \class{SecurityException} or invoked \cq.
This representation would allow \tool to verify the existence, order, and the comparison operators of the authorization checks.
However, for complex services (e.g., the Activity Manager) this representation was infeasible due to the large number of authorization checks.
Additionally, we found that without more complex context-sensitive analysis, \tool could not extract authorization checks involving implicit flows.
Therefore, we simplified our consistency analysis to only consider \emph{the existence} of an authorization check for an entry point.
This approach requires two reasonable correctness assumptions:
(1) all authorization checks have been placed and ordered correctly, and
(2) all \cps have the correct comparison operator. 
\tool cannot detect violations of these assumptions.

The existence of authorization checks is easily represented as a set; however some processing is required.
For each variable in a authorization check statement, \tool must substitute all possible values for that variable.
More specifically, for each \cps and \cqs statement (i.e., conditional or method invoke statement), \tool must generate a list containing the product of all the possible variables and the values for each variable.
To reduce redundant output, \tool only computes the product for \cqs that do not have a return value or whose return value is not used in a \cps.

For this expansion, \tool applies an inter-procedural def-use analysis to each variable used in a statement, thereby obtaining the set of all possible values for that variable from the entry point to this specific use site.
It then computes the product of these sets to achieve the complete set of authorization checks for a single statement.
If a variable is assigned a value from the return of a method call, \tool does not consider the possible return values of the method, but instead includes a reference to all the possible targets of the method call.
Similarly, if a variable is assigned a value from the field of a class or an array, \tool includes only a reference to the field or array instead of all the possible values that may be assigned to the field or array.
As such, while the list of values largely consists of constants, it may also contain references to fields, methods and arrays.
The resulting combination of all the iterations of values for each \cp and each required \cq of an entry point makes up the final set of authorization checks output by \tool.

The resulting set has the potential to be exponentially large.
To prevent this growth and to remove noise in \tool's output, we apply several simplifications to the authorization checks 
\begin{conference}
(see the extended version~\cite{acminer-extended}).
\end{conference}
\begin{extended}
(see Appendix~\ref{p1:app:appendix_authCheckSimplify}).
\end{extended}
These simplifications are designed to increase the number of authorization checks that are equivalent from a security standpoint and in no way affect the completeness of the authorization checks.

\subsection{Refining Authorization Checks}
\label{dsgn:refine}

The techniques in Section~\ref{sec:id-auth-check} identify \emph{potential} \cps; however, not all conditional statements are authorization checks.
Performing consistency analysis at this point would be infeasible due to the excessive noise in the data.
Therefore, \tool uses a one-time, semi-automated method to significantly reduce noise.

Our key observation is that Section~\ref{sec:mine-auth-check} over approximates authorization checks on the path from entry points to a thrown \class{SecurityException} or \cq.
From this over-approximation, \tool can generate a list of all the strings and fields used in the conditional statements, a list of the methods whose return values are used in the conditional statements, and the methods containing the conditional statements. 
These values can be manually classified as authorization-related or not.
The general refinement procedure is as follows:
(1) a domain expert filters out values not related to authorization;
(2) the refined list is translated into generalized expressions;
(3) \tool uses the generalized expressions to automatically filter out values not related to authorization;
(4) the generalized expressions are refined until the automatically generated list is close to that defined by the domain expert.
While creating generalized expressions is time consuming, they can be used to analyze multiple Android builds with minimal modifications.

The specific refinement procedure is divided into two phases:
(\S\ref{sec:refine-cq}) identifying additional \cqs, and (\S\ref{refine-cp}) refining \cp identification.

\begin{table}[t!]
\caption{\label{table:cqs-org}Initial List of \textit{\CQS}}
\vspace{-0.5em}
\scriptsize
\scalebox{1}{
\rowcolors{3}{black!10}{}
\begin{tabular}{|p{2.8cm}|p{3.2cm}|}
\hline
\multicolumn{1}{|>{{\centering\arraybackslash}}c|}{\textbf{Classes}} & \multicolumn{1}{>{{\centering\arraybackslash}}c|}{\textbf{Methods}} \tabularnewline \hline
\tttscript{Context} \newline  \tttscript{ContextImpl} \newline \tttscript{ContextWrapper} & \tttscript{enforcePermission} \newline \tttscript{enforceCallingPermission} \newline \tttscript{enforceCallingOrSelfPermission} \newline \tttscript{checkPermission} \newline \tttscript{checkCallingPermission} \newline \tttscript{checkCallingOrSelfPermission} \tabularnewline 
\tttscript{AccountManagerService} & \tttscript{checkBinderPermission} \tabularnewline 
\tttscript{LocationManagerService} & \tttscript{checkPackageName} \tabularnewline
\tttscript{IActivityManager} \newline \tttscript{ActivityManagerProxy} \newline \tttscript{ActivityManagerService} & \tttscript{checkPermission} \tabularnewline
\tttscript{ActivityManagerService} \newline \tttscript{ActivityManager} & \tttscript{checkComponentPermission} \newline \tttscript{checkUidPermission} \tabularnewline
\tttscript{IPackageManager} \newline \tttscript{IPackageManager\$Stub\$Proxy} \newline \tttscript{PackageManagerService} & \tttscript{checkUidPermission} \newline \tttscript{checkPermission} \tabularnewline
\hline
\end{tabular}
}
\end{table}

\subsubsection{Identifying Additional \CQS}
\label{sec:refine-cq}

\tool uses \cqs as indicators of the existence of authorization checks when no \class{SecurityException} is thrown.
Our initial list of \textit{\cqs}, shown in Table~\ref{table:cqs-org}, was very limited and only contained \pInitialContextQueries{} methods.
To expand this list we performed the following steps:
(1) run \tool as described in Section~\ref{sec:mine-auth-check} using the initial list of \textit{\cqs}, 
(2) from the marked conditional statements, generate a list of the methods containing these conditional statements and the methods whose return values are used in these conditional statements, 
(3) have a domain expert inspect the list and identify methods that match our definition of a \cq, and add these to our list of \cqs, and
(4) repeat steps 1$\rightarrow$3 until no new \cqs are added to the list.
For Android AOSP 7.1.1, this procedure took about 48 hours and increased the number of \textit{\cqs} to \pManualContextQueries{} methods. 

To make this list reusable, we translate it into a set of generalized expressions that describe \cqs across different Android versions.
The expressions consist of regular expressions and string matches for the package, class, and name of a method, and also include conditional logic. An example expression is shown in 
Figure~\ref{p1:fig:design_cqExprExample}.
Overall, we defined 49 generalized expressions to describe \cqs for Android AOSP 7.1.1, which took $<$10 hours. The expressions enabled \tool to identify an additional \pTotalDiffManualContextQueries{} methods, resulting in a total of \pTotalContextQueries{} \cqs. Please see 
\begin{conference}
the extended version~\cite{acminer-extended}
\end{conference}
\begin{extended}
Appendix~\ref{p1:app:appendix_repCQS}
\end{extended}
for details on the translation procedure and our website~\cite{acminer-website} for the expression-list. 

\begin{figure}
\centering
\resizebox{0.65\columnwidth}{!}{
\begin{tikzpicture}

\lstset{
    showspaces=false,
    showtabs=false,
    breaklines=false,
    showstringspaces=false,
    breakatwhitespace=true,
    tabsize=2,
    captionpos=b,
    commentstyle=\bfseries\color{gray},
    keywordstyle=\bfseries\color{blue},
    stringstyle=\color{red}\bfseries,
    basicstyle=\ttfamily\footnotesize,
    morekeywords={and,or,not,starts-with-package,regex-name-words,equal-package},
    morestring=[b]`,
    alsoletter={-,.},
    comment=[l]{//},
    emph=[1]{android.,com.android.,android.test},
    emphstyle=[1]{\color{red}\bfseries},
}          

\node {
\begin{lstlisting}  % Start your code-block

(and (or (starts-with-package android.) 
         (starts-with-package com.android.))
     (regex-name-words `^(enforce|has|check)\s
         ([a-z\s]+\s)?permission(s)?\b`) 
         (not (equal-package android.test))...))

\end{lstlisting}
};

\end{tikzpicture}
}
\vspace{-0.5em}
\caption{Example expressions to describe \cqs.}
\label{p1:fig:design_cqExprExample}
\end{figure}

\subsubsection{Refining \CP Identification}
\label{refine-cp}

To refine the over-approximation of authorization checks, \tool again uses a semi-automated method of refinement, this time for \cps.
The process begins by running \tool with the refined \cqs from Section~\ref{sec:refine-cq}.
The expert then inspects lists of strings, fields, and methods for the \emph{potential} \cps, removing those not related to authorization.

From our exploration, we discovered a number of different categories of \cps. Some we were aware of such as those involving \textit{UID}, \textit{PID}, \textit{GID}, \textit{UserId}, \textit{AppId}, and package name.
However, even within these categories, we discovered new fields, methods, and contexts in which checks are performed.
We also discovered previously unknown categories of \cps including those: (1) involving \ttt{SystemProperties}, (2) involving flags, (3) performing permission checks using the string equals method instead of the standard check permission methods, (4) checking for specific intent strings, and (5) checking boolean fields in specific classes.
Finally, we discovered that a significant amount of noise was generated by the conditional statements of loops and sanity checks such as \textit{null} checks. Using all of the information gained from the exploration of elements related to possible \textit{\cps}, we defined a filter that refines \cp and reduces the overall noise. 

Overall, the exploration took about 56 hours.
We defined a 41-rule filter in about 16 hours (see our website~\cite{acminer-website} for the actual filter and
\begin{conference}
the extended version~\cite{acminer-extended}
\end{conference}
\begin{extended}
Appendix~\ref{p1:app:appendix_definingCPSFilter}
\end{extended}
for the specification process).
The application of the filter for Android AOSP 7.1.1 reduced what \tool considered to be \cps from \pInitialControlPredicates{} to \pTotalControlPredicates{}. Such a significant reduction not only underscores the need for a filter but also makes the consistency analysis (Section~\ref{sec:const}) more feasible.

\subsection{Consistency Analysis}
\label{sec:const}

The final step of \tool is consistency analysis of authorization checks for each entry point.
In this paper, we perform consistency analysis of all entry points within a service.
However, the methodology may work across multiple services, or even across different OEM firmwares.
\tool performs consistency analysis by automatically discovering underlying correlative relationships between sets of authorization checks.
Specifically, \tool adopts a targeted approach for association rule mining by using constraint-based querying.
It then uses these association rules to predict whether an entry point's authorization checks are consistent.
The results are presented to a domain expert for review.

Figure~\ref{p1:fig:overview_association_rule} shows an example association rule generated by \tool from the code in Figure~\ref{p1:fig:overview_code}.
$X$ and $Y$ are sets of authorization checks found in entry points. The rule states that if an entry point has check(s) from the set $X$, then it must also have the check(s) in set $Y$.
\tool then uses these generated rules to identify potential vulnerabilities by reporting entry points that violate the learned rules (i.e., if an entry point has all of the checks in $X$, but it is missing checks in $Y$, then a violation occurs).
For instance, Figure~\ref{p1:tbl:overview_association_rule} shows the three entry points that match $X$ for the rule in Figure~\ref{p1:fig:overview_association_rule}, out of which \method{hasUserRestriction} fails to satisfy the rule (it does not contain \method{checkManageUsersPermission}). 
On closer inspection, we discovered that all three entry points either set or get information about user restrictions.
Moreover, the functionality of \method{hasUserRestriction} is nearly identical to \method{hasBaseUserRestriction}, which suggests \method{checkManageUsersPermission} is needed.
As seen in these examples, \tool allows an expert to only compare entry points that are close 
in terms of their authorization checks, which is
more precise than comparing all entry points to one another.

The remainder of this section discusses \tool's approach to efficiently discover these association rules and how \tool uses them to detect inconsistencies in authorization checks.

\begin{figure}
\centering
\begin{subfigure}[b]{.5\columnwidth}
  \centering
  \setlength{\abovecaptionskip}{6pt}
  \setlength{\belowcaptionskip}{0pt}
  \scalebox{0.9}{
\begin{tikzpicture}[
  auto,
  block/.style={
    rectangle,
    draw,
    fill=white,
    line width=0.4pt,
    inner sep=0
  },
  line/.style ={draw, line width=0.3mm, -latex,shorten >=0.35cm}
]

\draw node[block,label={\small \textbf{X}}] (X) {
    \scriptsize
    \begin{tabular}{l}
    \rule{0pt}{0.9\normalbaselineskip}\tttscript{isValidRestriction(String)}\\[0.1\normalbaselineskip]
    \end{tabular}
};

\draw (X) (0,-1.65) node[block,label={\small \textbf{Y}}] (Y) {
    \scriptsize
    \begin{tabular}{p{4cm}}
    \rule{0pt}{0.9\normalbaselineskip}\tttscript{hasManageUsersPermission()}\\[0.1\normalbaselineskip] \hline
    \rule{0pt}{0.9\normalbaselineskip}\tttscript{checkComponentPermission(MANAGE\_USERS,} \tttscript{getCallingUid(),} \tttscript{-1,} \tttscript{true)}\\[0.2\normalbaselineskip] \hline
    \rule{0pt}{0.9\normalbaselineskip}\tttscript{ROOT\_UID} \tttscript{==} \tttscript{getCallingUid()}\\[0.1\normalbaselineskip] \hline
   \rule{0pt}{0.9\normalbaselineskip}\tttscript{isSameApp(SYSTEM\_UID,} \tttscript{getCallingUid())}\\[0.1\normalbaselineskip]
    \end{tabular}
};

\draw[line] (X.south) +(0,0) -- (Y.north);

\end{tikzpicture}
}
  \caption{Association Rule}
  \label{p1:fig:overview_association_rule}
\end{subfigure}%
\begin{subfigure}[b]{.5\columnwidth}
  \centering
\scalebox{1}{
\scriptsize
\rowcolors{3}{black!10}{}
\begin{tabular}{|l|c|}
\hline
\multicolumn{0}{|c|}{\textbf{API}} & \multicolumn{1}{c|}{\textbf{X $\rightarrow$ Y}}
\tabularnewline \hline
\tttscript{hasBaseUserRestriction} & \checkmark 
\tabularnewline
\tttscript{setUserRestriction} & \checkmark 
\tabularnewline
\tttscript{hasUserRestriction} & \text{\sffamily X}
\tabularnewline \hline
\end{tabular}
}

  \caption{Entry Points For Rule}
  \label{p1:tbl:overview_association_rule}
\end{subfigure}
\vspace{-1em}
\caption{(a) shows an association rule generated from the code in Figure~\ref{p1:fig:overview_code} and (b) illustrates how the first 2 entry points satisfy the rule, while \ttt{hasUserRestriction} does not, indicating it contains one or more inconsistent authorization checks.} 
\end{figure}
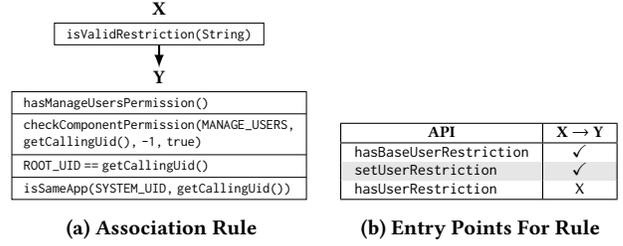

\subsubsection{Association Rule Mining}
\label{subsubsec:arm}

Association rule mining discovers correlative relationships between sets of authorization checks, $A = \{i_{1}, i_{2},\cdots,i_{n}\}$, across a set of entry points, $E=\{t_{1}, t_{2},\cdots, t_{m}\}$ where each entry point in $E$ contains a subset of the items in $A$.
 An association rule takes the form $X \implies Y$ where $X$ (antecedent) and $Y$ (consequent) are sets of authorization checks and $X$ and $Y$ are disjoint, i.e., $X \subseteq A$ and $Y \subseteq A$ and $X \cap Y = \emptyset$.

To avoid an excessive number of association rules, \tool uses two statistical constraints (support and confidence) to remove candidate association rules that are less than the thresholds minimum support (\textit{minsup}) and minimum confidence (\textit{minconf}).
Let $\alpha(I)$ represent the set of entry points in $E$ that contain the authorization checks $I \subseteq A$, i.e., $ \alpha(I) = \{ t \in E \mid \forall i \in I, i \in t \}$. The support of an association rule $X \implies Y$ is the probability that a set of authorization checks $Z = X \cup Y$ appears in the set of transactions $E$, i.e.,  $\sigma(Z) = \frac{|\alpha(Z)|}{|E|}$. The confidence of an association rule is an estimate of the conditional probability of the association rule $P(Y|X)$ where $X \implies Y$ and can be calculated as $\textit{conf}(X \implies Y) = \frac{\sigma(X \cup Y)}{\sigma(X)}$.

While association rule mining has been applied to similar problems by prior work~\cite{hxha10}, the large transaction size (i.e., number of authorization checks in an entry point) in our problem domain makes general association rule mining algorithms infeasible due to their exponential complexity. 
Therefore, \tool uses two main optimizations to reduce the complexity to polynomial time.

First, \tool only generates a subset of the association rules called closed association rules~\cite{szathmary06}. An association rule $X \implies Y$ is closed if $X \cup Y$ is a frequent closed itemset.  A frequent closed itemset is a set of authorization checks $C \subseteq A$ where the support of $C$ is greater than \textit{minsup} and there does not exist a superset $C'$ that has the same support as $C$. $C$ is closed iff $\beta(\alpha(C)) = C$ where $\beta(T)$ represents the largest set of common authorization checks present in the entry points $T$ where $T \subseteq E$, i.e., $\beta(T) = \{ i \in A \mid \forall t \in T, i \in t \}$. Note that only mining frequent closed itemsets is loss-less, because all frequent itemsets can be generated from the set of frequent closed itemsets, as proven by Zaki and Hsiao~\cite{ZH02}. Our proof that closed association rules also do not result in information loss can be found in
\begin{conference}
the extended version~\cite{acminer-extended}.
\end{conference}
\begin{extended}
Appendix~\ref{p1:app:appendix_proof}.
\end{extended}

Second, \tool generates closed association rules in a targeted manner by placing constraints on the authorization checks that appear in the antecedent of the association rule. Since the goal of consistency analysis is to predict whether an entry point's implementation of authorization checks is consistent, we are only interested in generating association rules where the antecedent of the association rule is a subset of the entry point's authorization checks (i.e., $X \subseteq A_{j}$ where $A_{j}$ is the authorization checks of entry point $e_{j}$). For example, consider $A_{j} = \{i_{1}, i_{2}, i_{3}\}$ and we have two frequent closed itemsets $\{i_{1}, i_{2}, i_{3}, i_{4}\}$ and $\{i_{5}, i_{6}, i_{7}\}$. The association rule $\{i_{1}, i_{2}, i_{3} \implies i_{4}\}$ is useful, as it could potentially hint that the authorization checks in $A_{j}$ is inconsistent and should also contain $i_{4}$. However, all of the association rules from the set $\{i_{5}, i_{6}, i_{7}\}$ do not provide additional information on the consistency of authorization checks in $A_{j}$, as the two sets are disjoint.

Further, assuming that the authorization checks that are present within an entry point are correct, we can force the antecedent to be constant. In particular, when generating association rules from a frequent closed itemset $I$ for an entry point $A_{j}$, we set $X = A_{j} \cap I$ and can generate association rules by varying the items in $Y$. If we reduce the authorization checks in $X$, then we are making the rule less relevant to the consistency of the entry point $A_{j}$ while keeping $X$ constant only produces the most relevant association rules.

\subsubsection{Inconsistency Detection and Output Generation}
\label{subsubsec:output}

\tool uses the \glspl{arule_p1} discussed in Section~\ref{subsubsec:arm}
to analyze the consistency of an entry point's authorization checks.
To minimize the amount of manual effort required to verify the presence of an inconsistency, we ensure the output only contains high confidence rules by setting \textit{minconf} to 85\%. Moreover, as we want the authorization checks in the consequent of an association rule to be formed by at least 2 entry points, we set the \textit{minsup} to $\frac{2}{|E|}$.

While generating the \glspl{arule_p1}, we mark an entry point's authorization checks as consistent if there exists a frequent closed itemset that contains the exact same authorization checks as the entry point, as this hints that the entry point's authorization checks are consistent with another entry point's authorization checks. In particular, entry point $e_{j}$'s authorization checks $A_{j}$ is consistent iff $\exists C \in A | C = A_{j} \wedge \beta(\alpha(C)) \wedge \sigma(C) \geq \frac{2}{|E|}$. In contrast, we mark an entry point's authorization checks as potentially inconsistent if an \gls{arule_p1} exists where the entry point's authorization checks are the antecedent of the rule and the consequent is not empty (i.e., $\exists X \implies Y | X \subseteq A_{j} \wedge Y \neq \emptyset$).

\tool outputs an HTML file for the domain expert to review
for each \gls{arule_p1} representing a potentially inconsistent entry point.
The HTML file contains the set of \glspl{supportac_p1} for the \gls{arule_p1} (i.e. the antecedent), the set of authorization checks being recommended by the \gls{arule_p1} (i.e. the consequent), and the 3 or more entry points that contain the authorization checks of the \gls{arule_p1}. This group of entry points can be subdivided into two sets: the \gls{target_p1} and the \glspl{support_p1}. The \gls{target_p1} is the entry point the \gls{arule_p1} has identified as being inconsistent, i.e., the entry point the \gls{arule_p1} is recommending additional authorization checks for. The \glspl{support_p1} are the 2 or more entry points where the \glspl{recommend_p1} occur. Note that the \glspl{supportac_p1} occur in both the target and the \glspl{support_p1}.
To aid the review, 
the HTML file also contains the set of all authorization checks from the \gls{target_p1} that do not occur in the \glspl{supportac_p1} and for all authorization checks, the method in the Android source code where the checks occur.

To reduce the manual effort required to confirm inconsistencies, we perform two post-processing techniques. First, we remove \glspl{arule_p1} where $|$\textit{\glspl{recommend_p1}}$| >= 5 * |$\textit{\glspl{supportac_p1}}$|$, since \glspl{arule_p1} that contain 100 authorization checks which imply 500 authorization checks is improbable. Second, we remove any remaining \glspl{arule_p1} that have over 100 \glspl{recommend_p1} as such \glspl{arule_p1} are unlikely to indicate inconsistencies, and the domain expert may not be able to evaluate such rules in a reasonable amount of time.

\section{Evaluation}
\label{sec:eval}

We evaluated \tool by performing an empirical analysis of the system services in AOSP version 7.1.1\_r1 (i.e., API 25)  built for a Nexus 5X device. Our analysis was performed on a machine with an Intel Xeon E5-2620 V3 (2.40 GHz), 128 GB RAM, running Ubuntu 14.04.1 as the host OS, OpenJDK 1.8.0\_141, and Python 2.7.6. 

We used \tool to mine the authorization checks of all the entry points from this build of AOSP and perform consistency analysis, as described in Section~\ref{sec:design}. Finally, we manually analyzed the inconsistencies using a systematic methodology to identify high risk (i.e., easily exploitable) and high impact vulnerabilities, and developed proof-of-concept exploits to validate our findings. Our evaluation is guided by the following research questions:
\begin{enumerate}[label=\textbf{RQ\arabic*},ref=\textbf{RQ\arabic*}]
  \item \label{rq:triaging} {\em Does \tool reduce the effort required by the domain expert in terms of the entry points that need to be analyzed?}
  \item \label{rq:vuln} {\em Do the inconsistencies identified by \tool help a domain expert in finding security vulnerabilities?}
  \item \label{rq:causes} {\em What are the major causes behind inconsistencies that do not resolve to security vulnerabilities?}
  \item \label{rq:prior} {\em Is \tool more effective than prior work at detecting vulnerabilities in system services?}
  \item \label{rq:time} {\em What is the time required by \tool to analyze all the system services in a build of Android?}
\end{enumerate}

We now highlight the salient findings from our evaluation, followed by the categorization of the discovered vulnerabilities. The categorization of non-security inconsistencies, developed via a systematic manual analysis of our results, is described in Section~\ref{p1:sec:eval_otherInconsistencies}.

\subsection{Evaluation Highlights}
\label{p1:sec:eval_summary}

\begin{figure}
\centering
\begin{subfigure}{.5\columnwidth}
  \centering
  \setlength{\belowcaptionskip}{-12pt}
  \setlength{\abovecaptionskip}{2pt}
  \scalebox{.50}{
\begin{tikzpicture}[
    title/.style={font=\huge\bfseries\ttfamily},
    num/.style={font=\LARGE\bfseries\ttfamily}
]
\draw[ultra thick,black,step=3cm] (0,0) rectangle (6,6);
\draw[ultra thick,black,step=3cm] (1,0) rectangle (6,4.5);
\draw[ultra thick,black,step=3cm] (2,0) rectangle (6,3);
\node[num,align=center] at (3,5.25) {\pEpTotal{} Total Entry Points};
\node[num,align=center] at (3.5,3.75) {\pEpWithLogic{} Entry Points with\\ Authorization Logic};
\node[num,align=center] at (4,2) {\pEpWithAtLeastOneRule{} Entry Points\\ with Association\\ Rule(s)};

\end{tikzpicture}
}
  \caption{Entry Point Reduction}
  \label{p1:fig:entry_points_reduce}
\end{subfigure}%
\begin{subfigure}{.5\columnwidth}
  \centering
  \setlength{\abovecaptionskip}{2pt}
  \setlength{\belowcaptionskip}{0pt}
  \scalebox{.50}{
\begin{tikzpicture}[
    title/.style={font=\huge\bfseries\ttfamily},
    num/.style={font=\LARGE\bfseries\ttfamily}
]
\draw[ultra thick,black,step=3cm] (0,0) grid (6,6);
\node[num,text width=2.5cm,align=center] at (1.5,4.5) {\pHighRiskHighImpactEp{} Entry Points};
\node[num,text width=2.5cm,align=center] at (4.5,4.5) {\pHighRiskEp{} Entry Point};
\node[num,text width=2.5cm,align=center] at (1.5,1.5) {\pHighImpactEp{} Entry Points};
\node[num,text width=2.5cm,align=center] at (4.5,1.5) {\pLowRiskLowImpactEp{} Entry Points};
\node[title,rotate=90] at (-0.5,4.5){High Risk};
\node[title,rotate=90] at (-0.5,1.5){Low Risk};
\node[title] at (1.5,6.5){High Impact};
\node[title] at (4.5,6.5){Low Impact};
\end{tikzpicture}
}
  \caption{Risk vs. Impact}
  \label{p1:fig:eval_risk_vs_impact}
\end{subfigure}
\caption{(a) shows how \tool is able to reduce the scope of the AOSP 7.1.1 system code a domain expert needs to evaluate and (b) breaks down the \pTotalRiskImpactEp{} vulnerabilities in terms of risk and impact.} 
\end{figure}
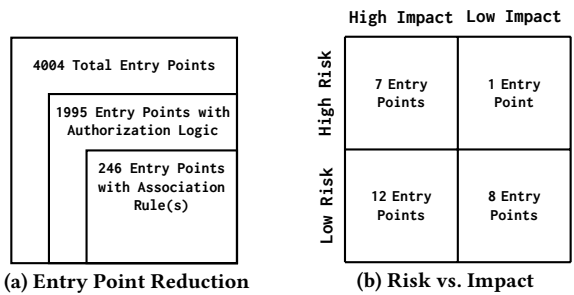

As shown in Figure~\ref{p1:fig:entry_points_reduce}, \tool reduces the total number of entry points that need to be manually analyzed down to just \pEpWithAtLeastOneRule{} entry points with inconsistent authorization checks, a \pEpReduction{}\% reduction (\ref{rq:triaging}). As a result, \tool significantly enhances a domain expert's ability to evaluate the consistency of access control enforcement in the Android system by minimizing the effort required. 
Further, \tool took approximately 1 hour and 16 minutes to mine the authorization checks of all entry points from the system image of the AOSP build, and spent an additional 30 minutes producing the HTML files for the \glspl{arule_p1} that represent potentially vulnerable entry points. While \tool could be optimized further, time taken by \tool is a feasible cost, given its scalability benefits over a fully manual analysis (\ref{rq:time}). 

On manually analyzing the \pEpWithAtLeastOneRule{} entry points, we discovered a total of \pTotalRiskImpactEp{} entry points containing security vulnerabilities (\ref{rq:vuln}). As Figure~\ref{p1:fig:eval_risk_vs_impact} illustrates, these \pTotalRiskImpactEp{} vulnerabilities were then classified in terms of their risk (i.e., the ease of exploiting a vulnerability) as well as the impact (i.e., the gravity of the consequence of an exploited vulnerability). Using this criteria, we found \pHighRiskHighImpactEp{} vulnerabilities that were high risk as well as high impact, \pHighRiskEp{} vulnerability that was high risk only, \pHighImpactEp{} vulnerabilities that were high impact only, and \pLowRiskLowImpactEp{} vulnerabilities that were low in terms of both risk and impact. All \pTotalRiskImpactEp{} vulnerabilities have been submitted to Google. So far, 2 of our vulnerabilities have been assigned a ''moderate'' Android Security Rewards (ASR) severity level, which is generally awarded to bypasses in access control mechanisms (e.g., restrictions on constrained processes, or general bypasses of privileged processes~\cite{asr}). In Section~\ref{p1:sec:eval_findings}, we categorize these \pTotalRiskImpactEp{} vulnerabilities according to their effect; however, we only discuss a few of these vulnerabilities in depth, due to space constraints.

\tool is significantly more effective than prior work at identifying inconsistent authorization checks. For instance, \tool is able to identify \pTotalContextQueries{} unique \cqs using the semi-auto\-mated approach described in Section~\ref{sec:refine-cq}, a drastic \pPrecentIncreaseContextQueries{}\% improvement over the original \pInitialContextQueries{} \cqs that encompass a majority of the \cqs considered by Kratos~\cite{socq+16}. Further, while AceDroid~\cite{ahs+18} is more comprehensive than Kratos in its identification of Android's authorization checks, it relies on a manually defined list of \cqs, which is insufficient. That is, as described in Section~\ref{sec:refine-cq}, our thorough attempts at identifying \cqs through manual observation alone resulted in the identification of only \pManualContextQueries{} \cqs, \pPrecentIdContextQueries\% of the total \cqs that \tool is able to find using its semi-automated approach. Thus, while AceDroid does not provide quantitative information on its set of \cqs, we can certainly say that it is not as complete as \tool in its identification of Android's authorization checks. 
Indeed, the \cq \ttt{isValidRestriction} in Figure~\ref{p1:fig:overview_code} is an example of a \cq that neither AceDroid nor Kratos was able to identify, and in fact, one that we missed in our manual definition of Android's authorization checks. However, through the general expressions, \tool was able to identify \ttt{isValidRestriction} as a \cq and the vulnerability outlined in Figure~\ref{p1:fig:overview_code}.  Moreover, neither AceDroid nor Kratos makes any mention of the App Ops restrictions in their definition of Android's authorization checks. Yet \tool is able to identify 2 vulnerabilities relating to the App Ops restrictions (see Section~\ref{p1:sec:eval_findings}). While a full empirical comparison with Kratos and AceDroid is infeasible due to the lack of source code access, our evaluation demonstrates that \tool makes significant advancements to existing work in terms of the coverage of the authorization checks, making the consistency analysis as complete as possible (\ref{rq:prior}).

Finally, \tool produced \pTotalAssociationRules{} \glspl{arule_p1} denoting inconsistent authorization checks in \pEpWithAtLeastOneRule{} entry points. Some entry points had more than one inconsistency. Furthermore, while some inconsistencies were indeed valid security vulnerabilities (\pTotalRiskImpactRule{}/\pTotalAssociationRules{}), others were a result of irregular coding practices in Android (\pTotalNonVulnerability{}/\pTotalAssociationRules{}) or indicative of \tool's limitations in terms of analyzing the semantics of the authorization checks (\ref{rq:causes}). The limitations identified via our analysis point to hard problems in analyzing Android's access control logic and motivate future work.


\subsection{Findings}
\label{p1:sec:eval_findings}

\begin{table*}[t]
\centering
\scriptsize
\caption{{Description of vulnerabilities, along with the services in which they are present}}
\label{tab:vuln}
\rowcolors{3}{black!10}{}
\def\arraystretch{1.2}
\begin{tabular}{|p{5cm}|p{12cm}|}
\hline
\multicolumn{0}{|c|}{\textbf{ Associated Entry Point (Service)}} & \multicolumn{1}{c|}{\textbf{Vulnerability Description}}
\tabularnewline \hline
\multicolumn{2}{c}{\textbf{VC1: Multi-user Enforcement}}
\tabularnewline \hline
1. \tttscript{getInstalledApplications} (\tttscript{PMS}) & Missing the \tttscript{enforceCrossUserPermission} check, allowing any app on one user profile to discover apps installed on other profiles.
\tabularnewline   
2. \tttscript{getPackagesHoldingPermissions} (\tttscript{PMS}) & Missing \tttscript{enforceCrossUserPermission}, allowing any app on one user profile to get sensitive permission information about other profiles.
\tabularnewline 
3. \tttscript{hasUserRestriction} (\tttscript{UMS}) & Missing the \tttscript{hasManageUsersPermission} check, which checks for the permission \tttscript{MANAGE\_USERS}, is missing, allowing any user to discover the restrictions on their own and other user profiles.
\tabularnewline
4. \tttscript{checkUriPermission} (\tttscript{AMS}) & Missing the \tttscript{handleIncomingUser} check that verifies if a user can operate on behalf of another, allowing any user access to content provider URIs belonging to another user, so long as the app making the request has access to the content provider.
\tabularnewline
5. \tttscript{grantUriPermission} (\tttscript{AMS}) & Missing the \tttscript{handleIncomingUser} check, with similar implications as \tttscript{checkUriPermission}.
\tabularnewline
6. \tttscript{killPackageDependents} (\tttscript{AMS}) & Missing the \tttscript{handleIncomingUser} check, allowing any user to kill the apps and background processes of another user.
\tabularnewline
7. \tttscript{setUserProvisioningState} (\tttscript{DPMS}) & Missing the \tttscript{enforceFullCrossUsersPermission} check, enabling any user to change another user profile's state.
\tabularnewline
8. \tttscript{setDefaultBrowserPackageName} (\tttscript{PMS}) & Missing the \tttscript{enforceCrossUserPermission} check, enabling any user to set the default browser of any other user. 
\tabularnewline
9. \tttscript{updateLockTaskPackages} (\tttscript{AMS}) & Missing \tttscript{handleIncomingUser}, enabling any user to modify the apps that may be permanently pinned to the screen in a kiosk like venue.
\tabularnewline
10. \tttscript{installExistingPackageAsUser} (\tttscript{PMS}) & Does not check if the target user exists, allowing any user to install apps on user profiles that may be created at a later time. 
\tabularnewline
11. \tttscript{setApplicationHiddenSettingAsUser} (\tttscript{PMS}) & Does not check if the target user exists, allowing any user to hide apps on user profiles that may be created in the future.
\tabularnewline
12. \tttscript{setAlwaysOnVpnPackage} (\tttscript{CS}) & Does not check for the \tttscript{no\_config\_vpn} user restriction, allowing a managed user to set its always on VPN to another application.
\tabularnewline
13. \tttscript{setWallpaperComponent} (\tttscript{WPMS}) & Missing the two checks \tttscript{isSetWallpaperAllowed} and \tttscript{isWallpaperSupported}, allowing a managed user to change their wallpaper.
\tabularnewline
14. \tttscript{startUpdateCredentialsSession} (\tttscript{ACMS}) & Missing checks \tttscript{canUserModifyAccounts} and \tttscript{canUserModifyAccountsForType}, allowing a user to trigger an update for the credentials of online accounts like Google and Facebook even when restricted. \setcounter{rownum}{0}
\tabularnewline \hline
\multicolumn{2}{c}{\textbf{VC2: App Ops}} 
\tabularnewline \hline
15. \tttscript{noteProxyOperation} (\tttscript{AOMS}) & Missing the \tttscript{verifyIncomingUid} check, which checks for a signature permission, allowing non-system apps to call this entry point.
\tabularnewline
16. \tttscript{getLastLocation} (\tttscript{LMS}) & A majority of the entry points in the \tttscript{LMS} use the \tttscript{AppOpsManager} check \tttscript{checkOp}, which is not intended for security, instead of the security check \tttscript{noteOp}. \tttscript{getLastLocation} uses the correct check. \setcounter{rownum}{0}
\tabularnewline \hline
\multicolumn{2}{c}{\textbf{VC3: Pre23}} 
\tabularnewline \hline
17. \tttscript{setStayOnSetting} (\tttscript{POMS}) & On systems with API 23 or above, the \tttscript{pre23} protection level allows any permission to be automatically granted to non-system apps built targeting the API 22 or below. This vulnerability allows non-system apps to access 6 additional entry points protected by the \tttscript{WRITE\_SETTINGS} \textit{signature} permission, as \tttscript{WRITE\_SETTINGS} also has the \tttscript{pre23} protection level. \setcounter{rownum}{0}
\tabularnewline \hline
\multicolumn{2}{c}{\textbf{VC4: Miscellaneous}} \tabularnewline \hline
18. \tttscript{unbindBackupAgent} (\tttscript{AMS}) & Missing check for if caller is performing a backup, allowing any app to disrupt the backup process of another app.
\tabularnewline
19. \tttscript{getPersistentApplications} (\tttscript{PMS}) & Missing system UID check, allowing any non-system app to discover what apps and services permantly run in the background.
\tabularnewline
20. \tttscript{logEvents} (\tttscript{MLS}) & Incorrect check for permission \tttscript{CONNECTIVITY\_INTERNAL}, should check \tttscript{DUMP} when writing sensitive data to logs.
\tabularnewline
21. \tttscript{getMonitoringTypes} (\tttscript{GHS}) & Missing check \tttscript{checkPermission}, allowing a caller access both fine and coarse levels of geofence location data.
\tabularnewline
22. \tttscript{getStatusOfMonitoringType} (\tttscript{GHS}) & Missing check \tttscript{checkPermission}, with similar implications as \tttscript{getMonitoringTypes}.
\tabularnewline
23. \tttscript{setApplicationEnabledSetting} (\tttscript{PMS}) & Missing \tttscript{isPackageDeviceAdmin} check, allowing an app to disable an active administrator app.
\tabularnewline
24. \tttscript{setComponentEnabledSetting} (\tttscript{PMS}) & Missing \tttscript{isPackageDeviceAdmin} check, allowing an app to disable components of an active administrator app.
\tabularnewline
25. \tttscript{convertFromTranslucent} (\tttscript{AMS}) & Missing check for \tttscript{enforceNotIsolatedCaller}, allowing a isolated process to affect the transparency of windows.
\tabularnewline
26. \tttscript{notifyLockedProfile} (\tttscript{AMS}) & Missing check for \tttscript{enforceNotIsolatedCaller}, allowing an isolated process to trigger a retrun to the home screen.
\tabularnewline
27. \tttscript{setActiveScorer} (\tttscript{NSS}) & Missing \tttscript{BROADCAST\_NETWORK\_PRIVILEGED} permission check which is always paired with the \tttscript{SCORE\_NETWORKS} permission check.
\tabularnewline
28. \tttscript{getCompleteVoiceMailNumberForSubscriber} (\tttscript{PSIC}) & Incorrect check for permission \tttscript{CALL\_PRIVILEGED} instead of the \tttscript{READ\_PRIVILEGED\_PHONE\_STATE} results in coarse-grained enforcement.\setcounter{rownum}{0}
\tabularnewline \hline
\multicolumn{2}{p{6.5in}}{AMS=ActivityManagerService; 
AOMS=AppOpsManagerService;
CS=ConnectivityService;
DPMS=DevicePolicyManagerService;
LMS=LocationManagerService;
PMS=PackageManagerService;
POMS=PowerManagerService;
UMS=UserManagerService;
WPMS=WallpaperManagerService;
PSIC=PhoneSubInfoController;
ACMS=AccountManagerService;
MLS=MetricsLoggerService;
GHS=GeofenceHardwareService;
NSS=NetworkScoreService}
%
%
%
\end{tabular}
\end{table*}

Table~\ref{tab:vuln} describes the vulnerabilities discovered through our analysis of Android 7.1.1 with \tool. On manually analyzing the inconsistent entry points produced by \tool, we discovered \pTotalRiskImpactEp{} entry points that represent security vulnerabilities. While most of these entry points represent one vulnerability each, two entry points (i.e., \ttt{getLastLocation} and \ttt{setStayOnSetting}, vulnerabilities 15 and 16 in Table~\ref{tab:vuln} respectively) each led us to clusters of multiple identically vulnerable entry points, as described later in this section. For simplicity, we count each cluster as a single vulnerability.

We group the vulnerabilities into the following 3 categories: {\sf (1)} user separation and restrictions, {\sf (2)} App Ops, {\sf (3)} and pre23.  This categorization is based on the subsystems affected by the vulnerabilities (e.g., App Ops), as well as the characteristics they have in common (e.g., pre23). Additionally, some vulnerabilities that are hard to classify have been categorized as {\sf (4)} miscellaneous. 

\myp{VC1: Multi-user Enforcement}
As shown in Table~\ref{tab:vuln}, a majority (i.e., \pUserEps{}) of the vulnerabilities affect Android's separation among users (i.e. user profiles in Android's multi-user enforcement~\cite{multi-user}). These can be further divided into four subcategories based on how they may be exploited: (1) leaking user information across users, (2) operating across users, (3) modifying user settings before a user exists, and (4) allowing users to bypass restrictions.

\emparagraph{1. Leaking Information to Other Users} In 5 entry points (i.e., 1$\rightarrow$5 in Table~\ref{tab:vuln}), the lack of checks leads to potential leaks of security-sensitive information to other users. For instance, using the vulnerable entry point \ttt{getInstalledApplications} in the \ttt{PackageManagerService}, any user can learn of the applications another user has installed, as the entry point does not enforce any checks other than checking if the user being queried exists.
Similarly, the entry point \ttt{hasUserRestriction} in the \ttt{UserManagerService}, previously used as the motivating example, is not protected with the \textit{signature} level permission \ttt{MANAGE\_USERS}, which is present in the similar \ttt{hasBaseUserRestriction} entry point. This omitted authorization check allows a user to know of the restrictions placed on other users, which is security-sensitive information that should not be public. The entry points \ttt{getPackagesHoldingPermissions}, \ttt{checkUriPermission} and \ttt{grantUriPermission} similarly leak sensitive information.

We experimentally confirmed the existence of both the vulnerabilities in \ttt{hasUserRestriction} and \ttt{getInstalledApplications}  in Android 7.1.1 as well as Android 8.1. We have submitted bug reports to Google and received ''moderate'' ASR severity level for both the bugs. Further, we confirmed that the vulnerability in \ttt{getPackagesHoldingPermissions} was fixed in Android 8.1. As a result, we could not submit it to Google's bug program, which only considers bugs affecting the latest version of Android. All remaining vulnerabilities have been reported to Google.

\emparagraph{2. Operating Across Users} Missing authorization checks in 4 entry points (i.e., 6$\rightarrow$9 in Table~\ref{tab:vuln}) allow users to bypass multi-user restrictions and perform sensitive operations on behalf of other users. For example, we discovered that the entry point {\ttt{killPackageDependents}} takes in a \textit{userId} as an argument but does not actually verify whether the calling user is allowed to perform operations on behalf of the supplied \textit{userId}. This allows {\em any user to kill the apps and background processes of any other user}. A similar flaw in entry point \ttt{setUserProvisioningState} enables any user to set the state of any other user profile to states such as "managed", "unmanaged", or "finalized". Such changes may be dangerous. For instance, a user may be able to set their managed enterprise profile to unmanaged, releasing it from the administrator's control.

Fortunately, all four entry points described in this category can only be called from apps installed on the system image (i.e., are protected by authorization checks that ensure this). This indirectly mitigates some damage, by making the vulnerabilities difficult to exploit from a third-party app. However, capability leaks in privileged apps may allow such vulnerabilities to be exploited by third-party apps, as prior work has demonstrated~\cite{gzwj12,wgz+13}. All of these vulnerabilities have been reported to Google.

\emparagraph{3. Modifying User Settings Before A User Exists}  Both the entry points \ttt{installExistingPackageAsUser} and \ttt{setApplicationHiddenSettingAsUser} do not perform the authorization check \ttt{exists}, which verifies if a the \textit{userId} passed into the entry points represents a valid user. Without this check, it is possible for a caller to install an app for a non-existent user or hide an app from a non-existent user. Thus, when the user for whom this change was made is actually created, these settings will already be in place. These entry points are only callable from systems apps; however, system apps may be compromised or may leak capabilities, and the \textit{exists} check needs to be in place to prevent system apps from being tricked into allowing users to install apps in profiles that have yet to be created (e.g., installing apps in a future enterprise profile). We have submitted these vulnerabilities to Google.

\emparagraph{4. Allowing Users to Bypass Restrictions} Vulnerabilities in entry points 12$\rightarrow$14 from Table~\ref{tab:vuln} allow a user to perform operations despite the restrictions placed on the user profile. For instance, the entry point \ttt{setAlwaysOnVpnPackage} does not check for the restriction \ttt{no\_config\_vpn}, allowing a managed user to set the always on VPN for the user profile to another application, effectively switching VPN connections. The entry points \ttt{setWallpaperComponent} and \ttt{startUpdateCredentialsSession} have similar vulnerabilities. All of these vulnerabilities have been reported to Google.

\myp{VC2: App Ops}
\tool identified weaknesses related to App Ops. 
One such vulnerability lies in the \ttt{noteProxyOperation} of the \ttt{AppOpsService}. The entry point makes a note of an application performing some operation on behalf of some other application through IPC. However, unlike other entry points in the \ttt{AppOpsService}, \ttt{noteProxyOperation} is missing the authorization check \ttt{verifyIncomingUid} which includes a check for the \textit{signature} level permission \ttt{UPDATE\_APP\_OPS\_STATS}. Without \ttt{verifyIncomingUid}, it is possible for any non-system app to use \ttt{noteProxyOperation} to query the restrictions a user has placed on other apps, thus retrieving information that should not be available to non-system apps. 

\begin{figure}
\centering
\resizebox{0.9\columnwidth}{!}{
\begin{tikzpicture}

\lstset{
    style=javaStyle,
    emph=[1]{checkOp},
    emphstyle=[1]{\color{blue}\bfseries},
    emph=[2]{String},
    emphstyle=[2]{\color{darkgreen}\bfseries},
    emph=[3]{op,uid,packageName},
    emphstyle=[3]{\color{darkorange}\bfseries},
    emph=[4]{CALL_PRIVILEGED},
    emphstyle=[4]{\color{Brown}\bfseries},
}          

\node {
\begin{lstlisting}  % Start your code-block

/** Do a quick check for whether an application might be 
 * able to perform an operation. This is not a security
 * check; you must use noteOp or startOp for your actual 
 * security checks, which also ensure that the given uid 
 * and package name are consistent. ... */
int checkOp(int op, int uid, String packageName) {...}

\end{lstlisting}
};

\end{tikzpicture}
}
\vspace{-0.5em}
\caption{The comment above \ttt{checkOp} from the class \ttt{AppOpsManager} that states it should not be used as a security check.
}
\label{p1:fig:eval_code_5}
\end{figure}


We discovered a set of identical vulnerabilities in App Ops through our analysis of the \ttt{getLastLocation} entry point in the \ttt{LocationManagerService}, which \tool pointed out as having inconsistent authorization checks. The \ttt{getLastLocation} entry point calls the authorization check \ttt{reportLocationAccessNoThrow} which performs the check \ttt{noteOpNoThrow} from the \ttt{AppOpsManager}, a wrapper for the \ttt{AppOpsService}. \tool correctly identified the use of \ttt{noteOpNoThrow} as an inconsistency since a majority of the entry points (9) in \ttt{LocationManagerService} use the authorization check \ttt{checkLocationAccess} which performs the check \ttt{checkOp} from the \ttt{AppOpsManager}. However, as Figure~\ref{p1:fig:eval_code_5} shows, the comment above the \ttt{checkOp} method clearly states that \ttt{checkOp} should not be used to perform a security check. Instead, one of the various forms of \ttt{noteOp} should be used. This implies that all 9 entry points using the \cq \ttt{checkLocationAccess} suffer from a vulnerability, and that the use of \ttt{reportLocationAccessNoThrow} in \ttt{getLastLocation} is actually appropriate. This instance demonstrates an interesting outcome of the use of consistency analysis in \tool. That is, our use of consistency analysis in \tool is also useful in identifying instances, where the majority of the related entry points are vulnerable. As described previously, for simplicity, we count this cluster of vulnerable entry points as a single vulnerability, which has been submitted to Google.  

\myp{VC3: Pre23}
\begin{figure}
\centering
\resizebox{0.9\columnwidth}{!}{
\begin{tikzpicture}

\lstset{
    showspaces=false,
    showtabs=false,
    breaklines=false,
    showstringspaces=false,
    breakatwhitespace=true,
    tabsize=2,
    captionpos=b,
    commentstyle=\bfseries\color{gray},
    keywordstyle=\bfseries\color{Plum},
    stringstyle=\color{red}\bfseries,
    basicstyle=\ttfamily\footnotesize,
    morekeywords={permission},
    morestring=[b]",
    alsoletter={-,",:},
    comment=[l]{//},
    emph=[1]{"},
    emphstyle=[1]{\color{red}\bfseries},
    emph=[2]{},
    emphstyle=[2]{\color{blue}\bfseries},
    emph=[3]{android:protectionLevel,android:name},
    emphstyle=[3]{\color{darkgreen}\bfseries},
}     

\node {
\begin{lstlisting}  % Start your code-block

<permission android:name="android.permission.WRITE_SETTINGS"
  android:protectionLevel="signature|preinstalled|appop|pre23" />

\end{lstlisting}
};

\end{tikzpicture}
}
\vspace{-0.5em}
\caption{The permission protection levels of \ttt{WRITE\_SETTINGS} in the \ttt{AndroidManifest.xml} file~\cite{pre23}}
\label{p1:fig:eval_code_7}
\end{figure}
%
\tool identified a group of vulnerabilities related to Android's \textit{pre23} permission protection level. 
The entry point \ttt{setStayOnSetting} in the \ttt{PowerManagerService} uses the authorization check \ttt{checkAndNoteWriteSettingsOperation}, which checks if an app has the \textit{signature} level permission \ttt{WRITE\_SETTINGS}. Permissions with the \textit{signature} protection level can only be granted to system apps (i.e., apps that were originally packaged with the system image). However, as shown in Figure~\ref{p1:fig:eval_code_7}, \ttt{WRITE\_SETTINGS} has an additional protection level of \textit{pre23}, which has an interesting effect on Android versions 6.0 or above (i.e., API 23 or above). It allows permissions marked as \textit{pre23} to be granted to non-system apps that target API 22 or below~\cite{protection_level}. Thus, as a result of the improperly defined permission protection levels for \ttt{WRITE\_SETTINGS}, the \ttt{pre23} grants non-system apps access to a \ttt{signature} level permission.

The damage resulting from the \textit{pre23} vulnerability is not restricted to the entry point \ttt{setStayOnSetting}. A simple search for the use of the permission \ttt{WRITE\_SETTINGS} in the authorization checks \tool mined for all entry points in the system revealed 13 additional entry points checking for the permission \ttt{WRITE\_SETTINGS}, 6 of which can be called from a non-system app using the \textit{pre23} vulnerability (i.e., these 6 entry points are not protected with any other \textit{signature} permission). Of the 6, the following 5 entry points deal with tethering and are located in the \ttt{ConnectivityService}: \ttt{setUsbTethering}, \ttt{stopTethering}, \ttt{startTethering}, \ttt{tether}, and \ttt{untether}. The last \ttt{setWifiApEnabled} was located in the \ttt{WifiServiceImpl} and allows a caller to set some WIFI access point configuration, causing the device to connect or disconnect from any WIFI access point the caller provides. These entry points are clearly more important to protect than \ttt{setStayOnSetting}, and an adversary may be able to do considerable damage by exploiting them. We do not count these entry points in our initial list of \pTotalRiskImpactEp{} vulnerabilities. All entry points affected by the \textit{pre23} vulnerability have been submitted to Google.

\myp{VC4: Miscellaneous Vulnerabilities}
\tool also identified \pOtherEps{} vulnerabilities related to information leaks, denial of service, disabling of administrator apps, and a mixture of other minor vulnerabilities. All of these vulnerabilities have been reported to Google.

\section{Non-security Inconsistencies}
\label{p1:sec:eval_otherInconsistencies}

\tool identified \pTotalOtherInconsistencies{} inconsistencies (i.e., rules) that did not represent vulnerabilities. Aside from the 20 rules that were caused by easily fixed bugs in \tool, we resolve these non-security inconsistencies to their likely causes, and classify them into 9 categories, shown in   
Table~\ref{p1:table:other_inconsistencies} (\ref{rq:causes}). The first three categories point to irregular coding practices, i.e., (1) inconsistent application of short-cuts to speed up access, (2) access bugs or discrepancies in how the permission should be used as per the documentation, or (3) inconsistent application of hard-coded checks that would potentially lead to vulnerabilities on future updates. The remaining 6 categories point to issues that could be corrected by engineering improvements to \tool, such as considering semantic equivalence between authorization checks, or the integration of call graph comparison and method-name comparison to mitigate the analysis of functionally different entry points. We provide additional details on all of the 9 categories in 
\begin{conference}
the extended version~\cite{acminer-extended}.
\end{conference}
\begin{extended}
Appendix~\ref{p1:app:appendix_eval_otherInconsistencies}.
\end{extended}

\begin{table}[t!]
\caption{\label{p1:table:other_inconsistencies}Non-security Inconsistencies}
\scriptsize
\scalebox{1}{
\rowcolors{3}{black!10}{}
\begin{tabular}{|l|c|} 
\hline
\multicolumn{1}{|>{{\centering\arraybackslash}}c|}{\textbf{Type of Inconsistency}} & \multicolumn{1}{>{{\centering\arraybackslash}}c|}{\textbf{Number of Rules}} \tabularnewline \hline
1. Shortcuts to Speed-Up Access & \pRulesSpeedUp{} \tabularnewline
2. Fixing Access Bugs & \pRulesIncreasedAccess{} \tabularnewline
3. Potential Vulnerabilities & \pRulesUnusedContextQueries{} \tabularnewline
4. Difference in Functionality & \pRulesDiffInFunctionality{} \tabularnewline
5. Checks With Different Arguments & \pRulesDiffArgs{} \tabularnewline
6. Noise in Captured Checks & \pRulesNoise{} \tabularnewline
7. Restricted to Special Callers & \pRulesSpecialCaller{} \tabularnewline
8. Semantic Groups of Checks & \pRulesPermissionCheck{} \tabularnewline
9. Equivalent Checks & \pRulesEquivalentChecks{} \tabularnewline
\hline
\end{tabular}
}
\end{table}

\section{Limitations}
\label{p1:sec:limitations}

While \tool is effective at discovering inconsistencies that lead to vulnerabilities, consistency analysis has a general limitation, i.e., it may not discover vulnerabilities that are consistent throughout code. Further, for precision, \tool does not handle the invocation of secondary entry points, i.e., calls to entry points from within other entry points. \tool omits the {\em ordering} of the authorization checks and hence does not identify improper operator use in \textit{\cps}, which we plan to explore in the future. Moreover, \tool's semi-automated analysis requires the participation of domain experts. However, as Section~\ref{sec:eval} demonstrates, our design significantly reduces manual effort in contrast with the manual validation of system services. As we have already analyzed AOSP version 7.1.1, only minor input is needed to analyze newer versions or vendor modifications.
Finally, \tool shares the general choices made by Android static analysis techniques for precision, i.e., it does not consider native code, or runtime modifications (e.g., reflection, dynamic code loading, Message Handlers).

\section{Related Work}
\label{sec:relwork}

\tool addresses a problem that has conceptual origins in prior work on
authorization hook validation for traditional systems.
Early investigations targeted the 
Linux Security Modules (LSM) hook placement in the Linux kernel, using techniques such as type analysis using CQUAL~\cite{zej02}, program
dominance~\cite{zjk04}, and dynamic analysis to create authorization graphs
from control flow traces~\cite{jez04,ejz02}.
As the lack of ground truth is a general
challenge for hook validation, prior work
commonly uses consistency analysis~\cite{jez04, ejz02,
tzm+08}. Closest to our work is AutoISES~\cite{tzm+08}, which infers
security specifications from code bases such as the Linux kernel and
Xen.  However, AutoISES assumes a known set of security functions or security-specific data structures, whereas \tool identifies a diverse set of authorization checks in the Android middleware. 


%

The closest to our approach is prior work on authorization hook validation in the Android platform, i.e., Kratos~\cite{socq+16} and AceDroid~\cite{ahs+18}. \tool distinguishes itself from Kratos and AceDroid through its deep analysis of Android's system services, and its significantly improved coverage of Android's authorization checks. 

Kratos~\cite{socq+16} compares a small subset of  Android's authorization checks across entry points of the same system image to look for inconsistent checks between different system services. \tool does not analyze for consistency across services. Instead, we hypothesize that entry points within a single service are similar in purpose, and hence, analyze the consistency of the authorization checks by comparing the entry points of every system service with other entry points in the same service. Further, \tool's semi-automated approach for identifying authorization checks results in a \pPrecentIncreaseContextQueries{}\% improvement over Kratos' manually-curated list (Section~\ref{sec:eval}).

Similarly, AceDroid~\cite{ahs+18} evaluates the consistency of the authorization checks among different vendor-modified Android images, and hence differs from \tool in terms of its objective. AceDroid makes key improvements over Kratos, as it considers various non-standard \cqs not considered by Kratos. However, AceDroid also relies on a manually-defined list of \cqs, which 
may produce only approximately \pPrecentIdContextQueries\% of the \cqs that \tool is able to find (Section~\ref{sec:eval}). 
\begin{extended}
While a quantitative comparison with Kratos and AceDroid is difficult due to the unavailability of code/data, this qualitative comparison demonstrates the remarkable improvements made by \tool's novel techniques.
\end{extended}
%

Finally, recent literature is rich with static and dynamic program analysis of
third-party Android apps targeted at privacy
infringement~\cite{egc+10, hhj+11, gcec12}, malware~\cite{zdyz14,
hzt+14, ejm+14}, as well as vulnerabilities~\cite{eomc11,
cfgw11}.  
As the Android platform and apps use similar programming
abstractions, researchers have applied these
tools and techniques to the platform code, e.g., for providing a mapping between APIs and corresponding permissions~\cite{fch+11, azhl12, bkml14} or automatically identifies privacy-sensitive sources and sinks~\cite{arb14}. Moreover, prior work has also studied the platform
code, to analyze OEM apps for capability leaks (e.g., Woodpecker~\cite{gzwj12} and
SEFA~\cite{wgz+13}), discover privilege escalation on update vulnerabilities (e.g., Xing et al.~\cite{xpw+14}), or uncover gaps in the file access control policies in OEM firmware images (e.g., Zhou et al.~\cite{zlz+14}). While \tool shares a similar objective, unlike prior work, \tool provides an
automated and systematic investigation of core platform services.




\section{Conclusion}
\label{sec:conc}
This paper provides an approach for the systematic and in-depth analysis of a crucial portion of Android's reference monitor, i.e., its system services. We design \tool, a static analysis framework that comprehensively identifies a diverse  array of authorization checks used in Android's system services, and then adapts the well-founded technique of association-rule mining to detect inconsistent access control among service entry points. We discover 28 security vulnerabilities by analyzing AOSP version 7.1 using \tool, and demonstrate significantly higher coverage of checks than the state of the art. Our work demonstrates the feasibility of in-depth analysis of Android's system services with \tool, as it significantly reduces the number of entry points that must be analyzed, from over 4000 with millions of lines of code to a mere \pEpWithAtLeastOneRule{}.

\myp{Acknowledgements} This work was supported by the Army Research Office (ARO) grant W911NF-16-1-0299 and the National Science Foundation (NSF) grants CNS-1253346 and CNS-1513690. Opinions, findings, conclusions, or recommendations in this work are those of the authors and do not reflect the views of the funders.

{\small
\bibliographystyle{ACM-Reference-Format}
\setcitestyle{numbers,sort&compress}
\citestyle{acmnumeric}
\bibliography{bibs/phone,bibs/ccs17,bibs/os,bibs/misc_links,bibs/add_to_os,bibs/analysis}
}

\begin{extended}
\appendix
\section{Extracting Java Class Files}
\label{p1:app:extracting_java_class}
\tool first obtains the Java class files for the desired version of Android.
Soot analyzes Java programs by translating their code into its own Java code representation (i.e., Jimple).
It parses \textit{.class}, \textit{.java}, \textit{.jimple}, and \textit{.jar} files. 
The recent incorporation of Dexpler~\cite{bklm12} into Soot adds support for \textit{.dex} files and \textit{.jar} files containing \textit{.dex} files (i.e., \textit{.apk} files).

\tool is designed to analyze both AOSP and OEM builds of the Android middleware.
As such, \tool must process Android's \textit{system.img} and extract the class files.
If the \textit{system.img} contains \textit{.odex} or \textit{.oat} files, \tool decompiles and recompiles them into \textit{.dex} using baksmali/smali~\cite{smali}.
As shown in Figure~\ref{fig:design_diagram}, \tool automates this process and outputs a \textit{.jar} file containing all the class files of the Android middleware in Soot's Jimple format.
This \textit{.jar} containing Jimple files is then then used on all subsequent runs of \tool.
This pre-processing can take upwards of 20 minutes and saves unnecessary re-computation.

\section{Extracting System Services and Entry Points}
\label{p1:app:extracting_java_class_2}

Android's middleware is largely composed of isolated services that only communicate through predefined \ttt{Binder} boundaries.
This division allows each service to be analyzed separately, and to define each service by the code reachable through its \ttt{Binder} entry points. \tool extracts system services and their entry points similar to PScout~\cite{azhl12}, Kratos~\cite{socq+16}, and AceDroid~\cite{ahs+18}.

Specifically, \tool uses Soot's class hierarchy information to find all classes that: (1) are subclasses of \ttt{android.os.Binder} and (2) implement the \ttt{onTransact} method.
Next, \tool identifies the signatures of the handler methods that \ttt{onTransact} methods call.
Our insight is that all such methods are invoked on the same receiver object (i.e., \textit{this}).
Note that handler methods need not be implemented in the \textit{stub} class itself: it can also be implemented in subclasses.
\tool therefore uses Soot's class hierarchy information to identify all of the \textit{stub's} subclasses.
Finally, the class and method information is used to identify all concrete methods in these \textit{stub} and subclasses that match the handler method signatures. These concrete methods become the entry points for analysis.

\section{The Exclude List}
\label{p1:app:appendix_excludeList}

\begin{table}[t!]
\caption{\label{p1:table:excluded_elements}Subset of the Excluded List for the Call Graph}
\scalebox{.95}{
\small
\rowcolors{3}{black!10}{}
\begin{tabular}{|p{2cm}|p{6.4cm}|}
\hline
\multicolumn{0}{|c|}{\textbf{Procedure}} & \multicolumn{1}{c|}{\textbf{Excluded Elements}} \tabularnewline \hline
Class Path & 
    \lword{java.*}, 
    \lword{javax.*}, 
    \lword{gov.nist.javax.*}, 
    \lword{org.*}, 
    \lword{sun.*}, 
    \lword{com.sun.*}, 
    \lword{com.ibm.*}, 
    \lword{com.google.common.*}, 
    \lword{soot.*}, 
    \lword{junit.*},
    \lword{com.android.dex.*}, 
    \lword{dalvik.*}, 
    \lword{android.test.*},
    \lword{android.text.*},
    \lword{android.util.*},
    \lword{android.animation.*},
    \lword{android.view.animation.*},
    \lword{android.icu.*},
    \lword{android.database.sqlite.*},
    \lword{android.content.res.*},
    \lword{com.android.org.bouncycastle.*},
    \lword{android.graphics.*},
    \lword{android.preference.*},
    \lword{android.os.UserHandle},
    \lword{android.os.Process},
    \lword{android.os.Binder},
    \lword{android.os.Debug},
    \lword{android.net.Uri},
    \lword{android.net.Uri\$*},
    \lword{libcore.util.Objects},
    \lword{android.os.Bundle},
    \lword{android.os.Parcel},
    \lword{android.view.View},
    \lword{libcore.io.IoUtils}
\tabularnewline 
Interface & 
    \lword{java.lang.Iterable},
    \lword{java.util.Iterator},
    \lword{java.util.ListIterator},
    \lword{java.lang.Comparable},
    \lword{java.util.Comparator},
    \lword{java.util.Collection},
    \lword{java.util.Deque},
    \lword{java.util.Enumeration},
    \lword{java.util.List},
    \lword{java.util.Map},
    \lword{java.util.Map\$Entry},
    \lword{java.util.NavigableMap},
    \lword{java.util.NavigableSet},
    \lword{java.util.Queue},
    \lword{java.util.Set},
    \lword{java.util.SortedMap},
    \lword{java.util.SortedSet},
    \lword{java.lang.Runnable},
    \lword{android.os.Parcelable},
    \lword{android.os.IInterface},
\tabularnewline
Interface All & 
    \lword{java.lang.AutoCloseable},
    \lword{libcore.io.Os},
    \lword{android.database.Cursor}
\tabularnewline
Superclass & 
    \lword{java.lang.Object},
    \lword{android.graphics.drawable.Drawable},
    \lword{android.content.Context}
\tabularnewline
Superclass All &
    \lword{com.android.internal.telephony.SmsMessageBase},
    \lword{java.lang.Throwable},
    \lword{android.app.Dialog},
    \lword{android.view.View}
\tabularnewline
Method Signature & 
    DevicePolicyManagerService: void writeToXml(XmlSerializer)\newline
    UserManagerService: void writeUserLP(UserData)\newline
    UserManagerService: void writeUserListLP()
\tabularnewline
\hline
\end{tabular}
}
\end{table}

A subset of the exclude list is displayed in Table~\ref{p1:table:excluded_elements} which groups the excluded elements by the five different procedures we use for excluding elements of the call graph (see our website~\cite{acminer-website} for the complete exclude list).
The first procedure \textit{Class Path} will exclude any method of a class whose full class name either exactly equals the listed name or starts with the listed name when the listed name ends with \ttt{.*} or \ttt{\$*}. The second \textit{Interface} excludes any method that implements a method of an interface that exactly equals the listed name. The third \textit{Interface All} excludes all methods of all classes that implement directly or indirectly an interface that exactly equals the listed name. The fourth \textit{Superclass} excludes any method that implements a method of a class that exactly equals the listed name. The fifth \textit{Superclass All} excludes all methods of all classes that extend directly or indirectly a class that exactly equals the listed name. Finally, \textit{Method Signature} excludes a single method that matches the supplied method signature.

To define the exclude list, we explored the call graph, looking for classes and methods that form roots of call graph areas not useful to our analysis. We started by broadly excluding, using the Class Path procedure, any packages that are not Android related as these are generic Java libraries and will not contain Android specific authorization checks. Since these generic Java libraries often contain interfaces and classes implemented and extended in code outside the Java libraries, we also found it necessary to use the Interface and Superclass procedures to cover subclasses of these in Android specific code. Next, as we still observed significant bloat, we broadened our search to include Android specific code. We started excluding packages such as \ttt{android.icu.*} and \ttt{com.android.org.bouncycastle.*} as these are common non-Android libraries that were integrated into the Android specific code but hold no Android-specific authorization checks. Lastly, we looked for methods in the call graph with a large number of outgoing edges and/or a large number of outgoing edges for the same method invoke statement, adding the methods and classes to the exclude list as needed. 

In addition to those described above, the exclude list has a few more additions and omissions. First, while \ttt{com.android.Context} is excluded using the Superclass procedure, all permission check methods in the class are omitted from this exclusion. This enables \tool to capture the authorization checks inside these various important methods. Second, classes such as \ttt{android.os.UserHandle} and \ttt{android.os.Process} have many  authorization check methods, yet are elements of the exclude list. Adding such classes to the exclude list gives context to the operations these authorization-checks perform. Such context would be lost otherwise since the actual checks involve manipulating integers values.

\section{Simplification of Authorization Checks}
\label{p1:app:appendix_authCheckSimplify}

We apply the following 10 simplification rules when generating authorization checks for a entry point and use the value \textit{ALL} in our output to indicate that a variable can be assigned any value. (1) Any variable that does not represent a variable of a primitive type or a string type is assigned the value ALL. We found such variables were safely ignored as they did not add any additional context to the authorization checks being output and instead generated a large amount of noise when included. (2) To handle loops and recursion, a value of ALL is assigned if a cycle is detected when resolving the possible values for a variable. (3) When a variable represents a parameter of the current entry point method, a value of ALL is assigned as the parameter's value could be anything. (4) A value of ALL is assigned to all variables obtaining there value from \textit{lengthof}, \textit{instanceof}, \textit{new}, and \textit{cast} expressions as no new context is gained about the authorization checks by including these expressions. (5) If an expression retrieving a value from an array has the value ALL for either its index or its array reference, then a value of ALL is assigned to the value retrieved from the array as no new context will be gained from such expressions. (6) Variables assigned from the return of a method in the class \ttt{Bundle} always get assigned the value ALL since \ttt{Bundle} is used to pass data through \ttt{Binder} and can therefore be anything. (7) For a \cps or \cqs when computing the product of all the possible variables and the values for each variable, if the possible set of values for a variable contains ALL then the set is transformed into a singleton set containing only the value ALL. (8) If a variable is assigned a value of \textit{NULL} and has any other value assigned to it, the NULL value is ignored so as to remove NULL checks. (9) When dealing with the final set of values generated from a \cp (i.e., a set of pairs), we remove a pair from this set if: the pair has NULL or ALL for either of its values, the pair has a constant value for both of its values, or both of the values in the pair are equivalent. (10) As Java object \ttt{equals} methods when used in a conditional statements take the form \ttt{if(o1.equals(o2)} \ttt{==} \ttt{0)}, the pair generated by \tool representing the authorization check would normally contain one entry for the value of the equals method and one entry for the constant value being compared. However, what is really desired is the values of the objects being compared by the equals (e.g., in the case of string comparisons). As such, we simplify the representation of such authorization checks by reconstructing the pairs so that they contain values for the calling object and the argument object (i.e., \ttt{o1} and \ttt{o2}) of the \ttt{equals} method.

\section{Representing \CQS}
\label{p1:app:appendix_repCQS}

\begin{table}[t!]
\caption{\label{p1:table:match-expr} Method and Field Matching Expressions}
\scalebox{.95}{
\small
\rowcolors{3}{black!10}{}
\begin{tabular}{|p{3.3cm}|p{5.1cm}|}
\hline
\multicolumn{0}{|c|}{\textbf{Operation}} & \multicolumn{1}{c|}{\textbf{Description}} \tabularnewline \hline
starts-with-(package\textbar class\textbar name) & Resolves to true if the package, class, or name starts with the given string for some method or field.\tabularnewline
ends-with-(package\textbar class\textbar name) & Resolves to true if the package, class, or name ends with the given string for some method or field.\tabularnewline
contains-(package\textbar class\textbar name) & Resolves to true if the package, class, or name contains the given string for some method or field.\tabularnewline
equals-(package\textbar class\textbar name) & Resolves to true if the package, class, or name equals the given string for some method or field.\tabularnewline
regex-(package\textbar class\textbar name) & Resolves to true if any part of the package, class, or name matches the given regular expression for some method or field.\tabularnewline
regex-name-words & Resolves to true if the method or field name matches any part of the given regular expression when the name is split at word boundary indicators.\tabularnewline
regex-class-words & Resolves to true if the method or field class matches any part of the given regular expression when the class is split at word boundary indicators. Note this operation takes in a second argument which specifies which class name to match in the case of an inner class. The indexing starts at 0 for the inner most class name and increases by 1 for each outer class name. An index of -1 means the regex can match any of the class names from inner most to outer most.\tabularnewline
\hline
\end{tabular}
}
\end{table}

\begin{figure}
\centering
\resizebox{0.95\columnwidth}{!}{
\begin{tikzpicture}

\lstset{
    showspaces=false,
    showtabs=false,
    breaklines=false,
    showstringspaces=false,
    breakatwhitespace=true,
    tabsize=2,
    captionpos=b,
    commentstyle=\color{green},
    keywordstyle=\bfseries\color{blue},
    stringstyle=\color{red}\bfseries,
    basicstyle=\ttfamily\footnotesize,
    morekeywords={and,or,not,starts-with-package,regex-name-words},
    morestring=[b]`,
    alsoletter={-,.},
    emph=[1]{android.,com.android.},
    emphstyle=[1]{\color{red}\bfseries},
}          

\node {
\begin{lstlisting}  % Start your code-block

(and 
  (or 
    (starts-with-package android.) 
    (starts-with-package com.android.)
  )
  (regex-name-words `^can\s(clear\sidentity|draw\soverlays
    |run\shere|user\smodify\saccounts|access\sapp\swidget
    |read\sphone\s(state|number)|caller\saccess\smock)\b`
  )
)

\end{lstlisting}
};

\end{tikzpicture}
}
\caption{An example of the expressions used to describe \cqs.}
\label{p1:fig:appendix_cqExprExample}
\end{figure}

\begin{figure}
\centering
\resizebox{0.95\columnwidth}{!}{
\begin{tikzpicture}

\lstset{
    style=javaStyle,
    emph=[1]{getProviderMimeType,enforceNotIsolatedCaller,canClearIdentity,checkComponentPermission,getCallingPid,getCallingUid,clearCallingIdentity,getUserId,isIsolated},
    emphstyle=[1]{\color{blue}\bfseries},
    emph=[2]{String,Binder,Uri,UserHandle,SecurityException},
    emphstyle=[2]{\color{darkgreen}\bfseries},
    emph=[3]{uri,userId,pid,uid,callingUid,callingPid,clearedIdentity,ident,caller},
    emphstyle=[3]{\color{darkorange}\bfseries},
    emph=[4]{INTERACT_ACROSS_USERS,INTERACT_ACROSS_USERS_FULL},
    emphstyle=[4]{\color{Brown}\bfseries},
}          

\node {
\begin{lstlisting}  % Start your code-block

public String getProviderMimeType(Uri uri, int userId) {
  enforceNotIsolatedCaller();
  int callingUid = Binder.getCallingUid();
  int callingPid = Binder.getCallingPid();
  boolean clearedIdentity = false;
  long ident = 0;
  if (canClearIdentity(callingPid, callingUid, userId)) {
    clearedIdentity = true;
    ident = Binder.clearCallingIdentity();
  }
  ...
}

void enforceNotIsolatedCaller(String caller) {
  if (UserHandle.isIsolated(Binder.getCallingUid()))
    throw new SecurityException();
}

boolean canClearIdentity(int pid, int uid, int userId) {
  if (UserHandle.getUserId(uid) == userId)
    return true;
  if (checkComponentPermission(INTERACT_ACROSS_USERS, pid, 
        uid, -1, true) == 0 
      || checkComponentPermission(INTERACT_ACROSS_USERS_FULL, pid,
        uid, -1, true) == 0)
    return true;
  return false;
}

\end{lstlisting}
};

\end{tikzpicture}
}
\caption{Pseudo-code from the \texttt{\protect\seqsplit{ActivityManagerService}} class.\protect\footnotemark}
\label{p1:fig:appendix_exampleCode1}
\end{figure}

\footnotetext{\label{ams2}\scriptsize\url{http://androidxref.com/7.1.1_r6/xref/frameworks/base/services/core/java/com/android/server/am/ActivityManagerService.java\#11401}}

To generalize our definition of \cqs, we developed a representation that allows us to express the \cqs as a combination of regular expressions and several string matching procedures used to describe the package, class, and/or method name, combined with conditional logic (i.e., \textit{and}, \textit{or}, and \textit{not}). Table~\ref{p1:table:match-expr} outlines a complete list of the base expressions used to define the regular expressions and string matching procedures. Figure~\ref{p1:fig:appendix_cqExprExample} shows an example expression from the set of expressions generated for Android AOSP 7.1.1 that describes, along with several others, the \cq \ttt{canClearIdentity} shown in line 19 of Figure~\ref{p1:fig:appendix_exampleCode1}.
In Figure~\ref{p1:fig:appendix_cqExprExample} we see the expression is actually made up of 4 different operations. As one would imagine, \textit{and} and \textit{or} operations behave just like their counterparts in normal conditional logic. Furthermore, \textit{starts-with-package} behaves as expected in that it checks to see if the package name of the class in which a method is defined starts with a given string, in this case either \ttt{android.} or \ttt{com.android.}. Finally, \textit{regex-name-words} is a bit more complicated in that checks to see if any part of a method name matches the given regular expression when the method name is split by word boundary indicators (i.e., splitting the name at capital letters in the case of camel-case or at an underscore character). That is, for the method named \ttt{canClearIdentity}, the method name is transformed into the string \ttt{can} \ttt{clear} \ttt{identity} before the regular expression matching is performed. By splitting the method name at word boundary indicators, this allows the regular expressions to be defined in terms of key words without having to worry about situations where one word might contain others in its spelling (e.g., \textit{is} and \textit{list}).

\section{Defining \CPS Filter}
\label{p1:app:appendix_definingCPSFilter}

The first step in defining the \cp filer was to have a domain expert explore the unique fields, strings, and methods (i.e., elements) used in every marked conditional statement to decide which of these elements are used in an authorization context and which are unimportant. In general, the exploration of these lists proceeds as follows: (1) Start by looking at the \cps contained within \cqs and the name of the \cqs themselves. As any element used within \cps of \cqs are likely important, we can already include these elements in our list of elements that indicate \cps. However, by studying the \cps, we can also gain further insight into what \cps might look like elsewhere in Android's code. Moreover, we can learn key words that might help with identifying additional elements with authorization context from our lists. (2) Perform a search on our lists of elements using the key words learned from the previous step as these can indicate common functionality. (3) Verify if the elements resulting from the previous search are actually used in \cps by looking at both how all conditional statements use an element and the overall flow of the methods containing the conditional statements using the element. Add any that are determined to be used in \cps to our list of elements that indicate \cps. Note, when recording elements that indicate \cps it is also important to indicate exactly how the elements are being used in \cps as certain elements may only be authorization related in specific context. For example, as shown in Figure~\ref{p1:fig:appendix_exampleCode2}, some \cps involve the use of methods that only when combined together in a specific way construct an authorization check (e.g., the string \ttt{equals} method and the \ttt{get} method of \ttt{SystemProperties} when \ttt{get} is provided the string \ttt{SYSTEM\_DEBUGGABLE}). (4) On the remaining unchecked elements, go through each element as in step 3. If new key words are added to our key word list as a result of finding a new indicator element, perform steps 2-3 again. Keep performing step 4 until all elements in our lists have been processed. 

In defining the \cp filter, it is important to understand that the use of the fields, methods, and strings identified do not always indicate \cps. Instead, they indicate \cps when used in a specific context. As such, any filter we design will have to allow us to specify such a context along with the fields, methods, and strings. Therefore, while we cannot rely solely on the expressions outlined in Table~\ref{p1:table:match-expr} we can reuse them. We use a customized XML document that \tool takes in as input as shown in Figure~\ref{fig:design_diagram} as our filter specification. The filter specification is constructed from a group of rules that are conditionally joined together using the \textit{and}, \textit{or}, and \textit{not} operators into one large conditional expression that specifies if a given conditional statement should be included in our final list of \cps.

\begin{figure}
\centering
\resizebox{0.95\columnwidth}{!}{
\begin{tikzpicture}

\lstset{
    showspaces=false,
    showtabs=false,
    breaklines=false,
    showstringspaces=false,
    breakatwhitespace=true,
    tabsize=2,
    captionpos=b,
    commentstyle=\color{green},
    keywordstyle=\bfseries\color{Plum},
    stringstyle=\color{red}\bfseries,
    basicstyle=\ttfamily\footnotesize,
    morekeywords={KeepFieldValueUse,Restrictions,IsInArithmeticChain},
    morestring=[b]`,
    alsoletter={-,",0},
    emph=[1]{0,","false"},
    emphstyle=[1]{\color{red}\bfseries},
    emph=[2]{and,or,not,starts-with-package,regex-name-words,regex-class-words},
    emphstyle=[2]{\color{blue}\bfseries},
    emph=[3]{UseUnion,HandleConstants,Value},
    emphstyle=[3]{\color{darkgreen}\bfseries},
}          

\node {
\begin{lstlisting}  % Start your code-block

<KeepFieldValueUse Value="(and 
    (regex-name-words `\b(flag(s)?)\b`) 
    (regex-class-words `\b((uri\spermission)
      |((package|application)\smanager\sservice)
      |permission\s(state|data)|package\ssetting
      |layout\sparams|display|(activity|application
      |provider|user|service|display|device)\sinfo)\b` 0)
    )">
  <Restrictions UseUnion="false">
    <IsInArithmeticChain HandleConstants="false"/>
  </Restrictions>
</KeepFieldValueUse>

\end{lstlisting}
};

\end{tikzpicture}
}
\caption{An example of the an entry in the \cp filter. This entry matches a number of flag field \cps, such as  Figure~\ref{p1:fig:appendix_exampleCode2} line 6, by their name, class, and use that have authorization context.}
\label{p1:fig:appendix_filterExample1}
\end{figure}

Figure~\ref{p1:fig:appendix_filterExample1} provides an example rule from our filter specification. It details the rule \textit{KeepFieldValueUse} used to match the \cp at line 6 of Figure~\ref{p1:fig:appendix_exampleCode2}. As shown, the rule uses the expressions outlined in Table~\ref{p1:table:match-expr} to first identify the use of a flag field in a conditional statement that may potentially indicate a \cp. The rule then further restricts what is considered a \cp by applying the restriction \textit{IsInArithmaticChain}. This restriction only allows conditional statements to be considered \cps if the field or method return value is used in a chain of arithmetic expressions whose resolution is then used in the conditional statement. The arithmetic operations are all standard Java binary and unary operators (e.g., +, -, ==, !=, and !). Such a restriction allows us to exclude situations like \verb| if(0 == method(flag))| where the flag field is being used in a conditional statement as an argument to a method while including conditional statements who use field and method return values as in Figure~\ref{p1:fig:appendix_exampleCode2} line 6.

\begin{figure}
\centering
\resizebox{0.95\columnwidth}{!}{
\begin{tikzpicture}

\lstset{
    showspaces=false,
    showtabs=false,
    breaklines=false,
    showstringspaces=false,
    breakatwhitespace=true,
    tabsize=2,
    captionpos=b,
    commentstyle=\bfseries\color{gray},
    keywordstyle=\bfseries\color{Plum},
    stringstyle=\color{red}\bfseries,
    basicstyle=\ttfamily\footnotesize,
    morekeywords={KeepFieldValueUse,Restrictions,IsInArithmeticChain,KeepMethodReturnValueUse,IsValueUsedInMethodCall,Matcher},
    morestring=[b]`,
    alsoletter={-,",0},
    comment=[l]{//},
    emph=[1]{0,","false","true",equals,"0","-1","StringMatcher","MethodMatcher"},
    emphstyle=[1]{\color{red}\bfseries},
    emph=[2]{and,or,not,starts-with-package,regex-name-words,regex-class-words,equal-name,regex},
    emphstyle=[2]{\color{blue}\bfseries},
    emph=[3]{UseUnion,HandleConstants,Value,class,Position},
    emphstyle=[3]{\color{darkgreen}\bfseries},
}          

\node {
\begin{lstlisting}  % Start your code-block

<KeepMethodReturnValueUse Value="(equal-name equals)">
  <Restrictions UseUnion="false">
    <IsInArithmeticChain HandleConstants="false"/>
    <Restrictions UseUnion="true">
      <IsValueUsedInMethodCall Position="-1">
        <Matcher class="MethodMatcher" Value="(and 
          (regex-class-words `\bsystem\sproperties\b` 0) 
          (regex-name-words `\bget\b`)
        )"/>
        <Restrictions UseUnion="false">
          <IsValueUsedInMethodCall Position="0">
            <Matcher class="StringMatcher" Value="(
              regex `ro\.(factorytest|test_harness
              |debuggable|secure)`
            )"/>
          </IsValueUsedInMethodCall>
        </Restrictions>
      </IsValueUsedInMethodCall>
      <IsValueUsedInMethodCall Position="0">
        ...
        // Same data as 
        // <IsValueUsedInMethodCall Position="-1">
        // except flipped
      </IsValueUsedInMethodCall>
    </Restrictions>
  </Restrictions>
</KeepMethodReturnValueUse>

\end{lstlisting}
};

\end{tikzpicture}

}
\caption{An example of the an entry in the \cp filter. This entry specifically matches situations such as the \cp shown in Figure~\ref{p1:fig:appendix_exampleCode2} line 5.}
\label{p1:fig:appendix_filterExample2}
\end{figure}

To further illustrate how rules are expressed in our filter specification, we take a look at the more complex example of Figure~\ref{p1:fig:appendix_filterExample2} which is a rule to match \cps like line 5 of Figure~\ref{p1:fig:appendix_exampleCode2}. The rule \textit{KeepMethodReturnValueUse} specification is similar to that of KeepFieldValueUse from our previous example except it has an additional set of restrictions. These restrictions are all forms of the \textit{IsValueUsedInMethodCall} which enables us to specify the possible arguments for a method as well as the possible calling object of the method. In this case, we have two sets of nested IsValueUsedInMethodCall restrictions the results of which are ored together when evaluated. The outer most IsValueUsedInMethodCall restriction in either set specifies that only the methods whose return value is used as part of the \ttt{equals} call, whose name contains \ttt{get}, and who is a member of the class \ttt{SystemProperties} be considered a \cp. The \textit{Position} attribute of IsValueUsedInMethodCall specifies where this restriction is to check for a value matching this description (i.e., -1 for the calling object and 0 for the first argument of the method). The inner IsValueUsedInMethodCall restriction of each set then further specifies that any method matching the outer restrictions description must also take as an argument at position 0 a string that matches the given regex. Combining all these restrictions together, we get a rule that says to treat any conditional statements as \textit{cps} if they are checking if a value \textit{A} equals some string where \textit{A} is the return value of a \ttt{get} method in \ttt{SystemProperties} retrieving the associated system value for some given key.

\begin{figure}
\centering
\resizebox{0.95\columnwidth}{!}{
\begin{tikzpicture}

\lstset{
    style=javaStyle,
    emph=[1]{dumpHeap,findProcess,checkCallingPermission,checkComponentPermission,getCallingPid,getCallingUid,getAppId,isIsolated,checkUidPermission,handleIncomingUser,isSameApp,equals,get},
    emphstyle=[1]{\color{blue}\bfseries},
    emph=[2]{String,SecurityException,ProcessRecord,SystemProperties,ApplicationInfo,Binder,UserHandle},
    emphstyle=[2]{\color{darkgreen}\bfseries},
    emph=[3]{userId,pid,uid,process,permission,exported,owningUid,proc,appId},
    emphstyle=[3]{\color{darkorange}\bfseries},
    emph=[4]{SET_ACTIVITY_WATCHER,SYSTEM_DEBUGGABLE,FLAG_DEBUGGABLE,ALLOW_FULL_ONLY,MY_PID},
    emphstyle=[4]{\color{Brown}\bfseries},
}          

\node {
\begin{lstlisting}  % Start your code-block

public boolean dumpHeap(String process, int userId, ...) {
  if(checkCallingPermission(SET_ACTIVITY_WATCHER) != 0)
    throw new SecurityException();
  ProcessRecord proc = findProcess(process, userId);
  if (!("1".equals(SystemProperties.get(SYSTEM_DEBUGGABLE, "0")))
      && 0 == (proc.info.flags & ApplicationInfo.FLAG_DEBUGGABLE))
    throw new SecurityException();
  ...
  return true;
}

ProcessRecord findProcess(String process, int userId) {
  int pid = Binder.getCallingPid();
  int uid = Binder.getCallingUid();
  userId = handleIncomingUser(pid, uid, userId, true, 
    ALLOW_FULL_ONLY, null);
  ...
  return proc;
}

int checkCallingPermission(String permission) {
  int pid = Binder.getCallingPid();
  int uid = UserHandle.getAppId(Binder.getCallingUid());
  return checkComponentPermission(permission, pid, uid,
      -1, true);
}

int checkComponentPermission(String permission, int pid, int uid, 
    int owningUid, boolean exported) {
  if(pid == MY_PID)
    return 0;
  int appId = UserHandle.getAppId(uid);
  if(appId == 0 || appId == 1000 
      || UserHandle.isIsolated(uid) 
      || (owningUid >= 0 && UserHandle.isSameApp(uid, owningUid)))
    return 0;
  if (!exported)
    return 1;
  return checkUidPermission(permission, uid);
}

\end{lstlisting}
};

\end{tikzpicture}
}
\caption{Pseudo-code from the \texttt{\protect\seqsplit{ActivityManagerService}} class.\protect\footnotemark}
\label{p1:fig:appendix_exampleCode2}
\end{figure}

\footnotetext{\label{ams1}\scriptsize\url{http://androidxref.com/7.1.1_r6/xref/frameworks/base/services/core/java/com/android/server/am/ActivityManagerService.java\#21497}}

Aside from rules like those presented above, the filter also covers a few corner cases. Mainly \cqs and loop conditionals. As defined in Section~\ref{sec:id-auth-check}, any conditional statement in a \cq or using a \cq's return value should be considered a \cp. As such, the filter should always include these conditional statements as \cps. The filter uses the description of \cqs (see Section~\ref{sec:refine-cq}) to identify \cqs and preserve conditional statements as \cps if a conditional statement is within the body of a \cq or the \cq's return value is used in an chain of arithmetic expressions.
Moreover, as mentioned above, one of the main sources of noise in \tool's view of the authorization checks was loop conditionals (i.e., the conditional statements that decide if the control flow should exit a loop). As loops conditionals are not important when viewing Android's authorization checks the filter explicitly rejects all conditional statements that are loop conditionals.

\section{Completeness Proof for Closed Association Rule Mining}
\label{p1:app:appendix_proof}

\begin{lemma}
If $X \implies Y$ is a closed association rule and not confident (i.e., $\textit{conf}(X \implies Y) < minconf$), there does not exist an itemset $Y' \subset Y$ where $X \implies Y'$ is a confident association rule (i.e., $\textit{conf}(X \implies Y') \geq minconf$) unless $X \cup Y'$ is also a closed frequent itemset.
\end{lemma}

\begin{proof}
Let $\textit{conf}(X \cup Y) = \frac{\sigma(X \cup Y)}{\sigma(X)}$ and $\textit{conf}(X \cup Y') = \frac{\sigma(X \cup Y')}{\sigma(X)}$ where $Y' \subset Y$. If $\textit{conf}(X \implies Y') \geq \textit{conf}(X \implies Y)$, then $\sigma(X \cup Y') \geq \sigma(X \cup Y)$. Due to the monotonicity property, $\sigma(X \cup Y') \geq \sigma(X \cup Y)$, because $Y' \subset Y$ and $\sigma(Y') \geq \sigma(Y)$. If $\sigma(X \cup Y') = \sigma(X \cup Y)$, then $\textit{conf}(X \cup Y') < minconf$. If $\sigma(X \cup Y') > \sigma(X \cup Y)$, then $X \cup Y'$ is also a closed itemset according to Lemma~\ref{citemset-lemma}.
\end{proof}

\begin{lemma}
    \label{citemset-lemma}
    If $X \cup Y$ is a frequent closed itemset, then $X \cup Y'$ where $Y' \subset Y$ is also a frequent closed itemset if $\sigma(X \cup Y') > \sigma(X \cup Y)$.
\end{lemma}
\begin{proof}
    Let $\sigma(X \cup Y') > \sigma(X \cup Y)$ and $\alpha(X \cup Y) \neq \emptyset$. For $X \cup Y'$ to not be a frequent closed itemset, then $\exists Y'' \supset Y' \mid \sigma(X \cup Y') = \sigma(X \cup Y'')$ by the definition of a frequent closed itemset. If $\sigma(X \cup Y'') = \sigma(X \cup Y')$, then $\alpha(X \cup Y'') = \alpha(X \cup Y')$. Further, since $\sigma(X \cup Y') > \sigma(X \cup Y) \wedge \alpha(Y) \neq \emptyset$, then $\alpha(Y) \subset \alpha(Y'')$. If $\alpha(Y) \not\subset \alpha(Y'')$, then $\sigma(X \cup Y'') < \sigma(X \cup Y')$ by definition.
    Since $Y \supset Y'$ exists and $\sigma(X \cup Y'') = \sigma(X \cup Y')$, then $\sigma(X \cup Y \cup Y'') = min(\sigma(X \cup Y), \sigma(X \cup Y''))$ by the monotonicity property. Since $\sigma(X \cup Y') > \sigma(X \cup Y) \wedge \sigma(X \cup Y') = \sigma(X \cup Y'') \implies \sigma(X \cup Y \cup Y'') = \sigma(X \cup Y)$. Therefore, $\nexists X \cup Y'' \mid \sigma(X \cup Y'') = \sigma(X \cup Y')$, because $X \cup Y$ is defined as a closed itemset. Thus, $X \cup Y'$ is also a frequent closed itemset.
\end{proof}

\section{Non-security Inconsistencies}
\label{p1:app:appendix_eval_otherInconsistencies}

\tool identified \pTotalOtherInconsistencies{} inconsistencies (i.e., rules) that did not represent vulnerabilities. Aside from the 20 rules that were caused by easily fixed bugs in \tool, we resolve these non-security inconsistencies to their likely causes, and classify them into 9 categories, shown in   
Table~\ref{p1:table:other_inconsistencies} (~\ref{rq:causes}). 

\subsection{Irregular Coding Practices}
\label{p1:sec:eval_otherInconsistencies_codingPractices}

In addition to detecting vulnerabilities, \tool can also be useful for detecting irregular coding practices in Android's system services. This may not only improve code quality, but may also help increase computation speed, fix access bugs, or indicate locations of future vulnerabilities, as described below.

\myp{1. Shortcuts to Speed-Up Access} As Table~\ref{p1:table:other_inconsistencies} shows, \tool detected \pRulesSpeedUp{} inconsistencies that when fixed will improve the performance of the system. For example, consider the \ttt{BatteryStatsService}. In this service, whenever an entry point checks for the permission \ttt{UPDATE\_DEVICE\_STATS}, the entry point first calls \ttt{enforceCallingPermission}, which contains an additional \textit{PID} check that grants the service quick access to its own entry point, bypassing the permission check. This optimization speeds up access without a significant security-cost. Now consider another entry point \ttt{takeUidSnapshots} in the same service.
\tool recommended this same \textit{PID} check to be included when other permissions in the service \ttt{BATTERY\_STATS} permission is checked for this entry point, i.e., ensuring the consistent application of such valuable short cuts.

\myp{2. Fixing Access Bugs} \tool detected \pRulesIncreasedAccess{} inconsistencies that when fixed solve access-related bugs, i.e., discrepancies in how a permission should be used (i.e., as per the documentation). For example, consider the entry point \ttt{getUserCreationTime} of the \ttt{UserManagerService}. The entry point is currently limited to being called by either the same user as the user indicated by the \text{userId} passed in as an argument or a parent of user. However, other similar entry points in the service also grant access to callers with the \textit{signature} level permission \ttt{MANAGE\_USERS}. Thus, \tool recommends the addition of the \ttt{MANAGE\_USERS} permission check to this entry point. This is consistent with the documentation, 
which states that callers with the \ttt{MANAGE\_USERS} permission may call \ttt{getUserCreationTime}.

\myp{3. Potential Vulnerabilities} \tool detected \pRulesUnusedContextQueries{} inconsistencies which may lead to future vulnerabilities. Consider the entry point \ttt{deletePackage} in the \ttt{PackageManagerService}, which checks for the permission \ttt{INTERACT\_ACROSS\_USERS\_FULL} when verifying if the calling user can operate on the user represented by the \textit{userId} argument. Almost all the entry points in the service use the authorization check \ttt{enforceCrossUserPermission} to perform a similar user-related authorization check. That is, \ttt{deletePackage} still performs the hard-coded permission check, and is inconsistent with the majority that call the modular \ttt{enforceCrossUserPermission} permission check. Thus, there is as strong possibility that the hard-coded check in \ttt{deletePackage} may be overlooked when additional user-related enforcement is introduced, resulting in a vulnerability.

\subsection{Improvements to \tool's Accuracy}
\label{p1:sec:eval_ImproveAccuracy}

Our systematic investigation of \tool's results leads to 6 causes of non-security rules that motivate future work:

\myp{1. Difference in Functionality} A majority of the non-security inconsistencies (i.e., 189) were caused in cases where the target and \glspl{support_p1} contained unrelated protected operations, and thus, comparing their authorization checks resulted in unusable \glspl{arule_p1}. We are exploring techniques such as call graph comparison and method-name comparison to improve entry point groupings to mitigate such rules.

\myp{2. Checks With Different Arguments} We observed 66 non-security inconsistencies where same authorization checks were instantiated in code with slightly different arguments, without affecting the security context.
This problem may be mitigated by making \tool's analysis more precise, i.e., by analyzing consistency in terms of the relevant arguments  and variables that actually affect the security context of the authorization check. Such fine-grained analysis is a non-trivial problem for future work.

\myp{3. Noise in Captured Checks} Despite the use of the \cp filter, \tool still identifies some conditional statements and method calls improperly as authorization checks. We have already determined the statements causing this issue, and are addressing it via a routine refinement of the \cp filter as well as the expressions used to identify \cqs.

\myp{4. Restricted to Special Callers} We found that numerous entry points in the system are restricted to being called by special callers (e.g., the \textit{UID} of system, shell, or root). As a result, any \gls{arule_p1}s generated for such entry points are valueless since the \gls{target_p1} is more restrictive than the \glspl{support_p1}. 
We are exploring the integration of a unified view of the hierarchy among the different authorization checks in \tool (i.e., in terms of which checks supersede others) to mitigate such issues.

\myp{5. Semantic Groups of Checks} We observed that a number of unrelated permission checks are always accompanied by checks for the system or root \textit{UID} or isolated processes. Thus, \tool ends up generating rules using the UIDs/isolated process checks as \glspl{supportac_p1}, recommending unrelated permissions as the missing authorization checks. That is, since \tool does not consider such semantic groups, it generates multiple incorrect rules in cases where a few entry points check for different, unrelated, permissions, while also checking for UIDs/isolated processes. From our analysis of the results, we have discovered that such rules can be filtered out as the generally follow a pattern.

\myp{6. Equivalent Checks} \tool does not consider the semantic equivalence between authorization checks, and thus, cannot eliminate \glspl{arule_p1} generated due to multiple checks having the same outcome. Fortunately, we have discovered numerous sets of equivalent checks through this analysis, which we plan to apply to \tool to improve its accuracy. 


\end{extended}

\end{document}